\documentclass[onecolumn]{article}

\usepackage{comment}
\usepackage[utf8]{inputenc} % allow utf-8 input
\usepackage[T1]{fontenc}    % use 8-bit T1 fonts
\usepackage[colorlinks=true,
linkcolor=red,
urlcolor=blue,
citecolor=blue]{hyperref}
\usepackage{url}            % simple URL typesetting
\usepackage{booktabs}       % professional-quality tables
\usepackage{amsfonts}       % blackboard math symbols
\usepackage{nicefrac}       % compact symbols for 1/2, etc.
\usepackage{microtype}      % microtypography
\usepackage{xcolor}         % colors
\usepackage{dsfont}
\newcommand\independent{\protect\mathpalette{\protect\independenT}{\perp}}
\def\independenT#1#2{\mathrel{\rlap{$#1#2$}\mkern2mu{#1#2}}}
 
\usepackage{amsmath}
\newcommand{\ie}{\emph{i.e.}} 
\newcommand{\eg}{\emph{e.g.}}

\usepackage{lipsum}
\usepackage{tikz}
\usetikzlibrary{bayesnet}
\usetikzlibrary{arrows}
\usepackage{sidecap}
\usepackage[export]{adjustbox}
\usepackage{bbm}
\usepackage{caption}
\usepackage{subcaption}
\usepackage[normalem]{ulem}

\usepackage{algorithm,balance}
\usepackage{algpseudocode}
\usepackage[round]{natbib}
\usepackage{amsthm}
\usepackage{hyperref}
\usepackage{varioref}
\usepackage{xr-hyper}
\usepackage{hyperref}
\begin{document}
	
	\title{Testing Granger Non-Causality in Panels with Cross-Sectional Dependencies}
	\author{Lenon Minorics, Caner Turkmen, David Kernert,\\ Patrick Bloebaum, Laurent Callot, Dominik Janzing \\ 
		{ \small  Amazon Research} \\
		{\small \{minorics, atturkm, davdkern, bloebp, lcallot, janzind\}@amazon.com }}
	
	\maketitle
	%\twocolumn[
	
	%\aistatstitle{Testing Granger Non-Causality in Panels with Cross-Sectional Dependencies}
	
	%\aistatsauthor{ Lenon Minorics  \And Caner Turkmen \And David Kernert }
	%\aistatsaddress{ Amazon Research \\ minorics@amazon.com \And Amazon Research \\ atturkm@amazon.com  \And Amazon Reserach \\ davdkern@amazon.com }
	%\aistatsauthor{ Patrick Bloebaum  \And Laurent Callot \And Dominik Janzing }
	%\aistatsaddress{Amazon Research \\ bloebp@amazon.com \And Amazon Reserach \\ lcallot@amazon.com \And Amazon Research \\ janzind@amazon.com }
	%\runningauthor{Lenon Minorics, Caner Turkmen, David Kernert, Patrick Bloebaum, Laurent Callot, Dominik Janzing}]
	\newtheorem{thm}{Theorem}
	\newtheorem{prop}[thm]{Proposition}
	\begin{abstract}
		This paper proposes a new approach for testing Granger non-causality on panel data. Instead of aggregating panel member statistics, we aggregate their corresponding p-values and show that the resulting p-value approximately bounds the type I error by the chosen significance level even if the panel members are dependent. We compare our approach against the most widely used Granger causality algorithm on panel data and show that our approach yields lower FDR at the same power for large sample sizes and panels with cross-sectional dependencies. Finally, we examine COVID-19 data about confirmed cases and deaths measured in countries/regions worldwide and show that our approach is able to discover the true causal relation between confirmed cases and deaths while state-of-the-art approaches fail.
	\end{abstract}
	\section{INTRODUCTION}
	Within the last decade, there has been growing awareness that causal inference can improve scientific research in many disciplines as
	interpretability and robustness become increasingly important~\citep{doshi2017towards, roscher2020explainable, marcinkevivcs2020interpretability, moraffah2020causal}. Causality is a crucial factor for gaining insights into the decision process of algorithms, which has many use cases such as avoiding bias and discrimination~\citep{mehrabi2019survey}, improving user experience~\citep{zhou2007understanding} and gathering biological insights~\citep{angermueller2016deep}. 
	
	If the causal relation between variables is known, causality can be used to study the interaction between statistical units such as estimating the average effect of treatments \citep{Imbens2015, Holland86}, analyze their mediation \citep{Berzuini2012}, detect the root causes of anomalies \citep{Janzing2019} or quantifying the causal influence of variables in a system \citep{Janzing2013, Janzing2020}.
	However, knowing the causal relations between variables of interest {\em a priori} is important for such applications.
	Unfortunately, these causal relations can often only be recovered from observational data under strong assumptions, even if hidden common causes (\ie, confounders) do not exist~\citep[Proposition 4.1]{Peters2017}.
	This leads to a seemingly unsolvable problem if further information about the data generating process is unavailable.
	For i.i.d. data, one way to mitigate this ambiguity is to impose additional assumptions about the generative process, such as linear non-Gaussian noise (LiNGAM) \citep{shimizu2006} or non-linear additive noise \citep{Hoyer08}. 
	For time series data, however, consequential additional information about the data generating process is fortunately available due to the time order. 
	That is, the present cannot causally influence the past, and hence the causal order is known.
	In many cases, this information is sufficient to recover the causal direction among variables in the form of time series \citep[Section 10.3.3]{Peters2017}.
	
	A well-known approach to causal discovery in time series data was developed by Granger \citep{Granger69, Granger80, Granger2003}. 
	Granger-causality is based on a simple definition: time series $\{X_t\}$ is said to Granger-cause time series $\{Y_t\}$ if the past of $\{X_t\}$ improves the prediction of $\{Y_t\}$ given its own past. 
	The presence of Granger-causality implies conditional dependence between the ``present'' $Y_t$ and the past of $X_t$, given the past of $Y_t$. By assuming there are no hidden common causes, we can deduce from Reichenbach's principle of common causes \citep{Reichenbach99} and the known causal order, that $\{X_t\}$ causally influences $\{Y_t\}$. In this definition of Granger-causality, we restricted to the bi-variate case for simplicity, however note that Granger defined the concept for general multivariate settings.
	
	Although many works have considered new approaches in time series causal discovery \citep{Shajarisales2015, Peters2012, Hyvaerinen2010, Kawahara2011}, procedures based on Granger-causality are still state-of-the-art due to their simplicity and favorable applicability in practice \citep{Berzuini2012, Bressler2011}. That is, as stated in \cite{Bressler2011} Chapter 22 Section 2.3,  Granger Causality does not rely on specific assumptions about the data generating process (except the exclusion of instantaneous effects, if instantaneous effects exist, additional assumptions have to be made) and is therefore particularly convenient for empirical investigations. 
	
	Several researchers considered an extension of classical causal discovery in time series to panel data \citep{Konya2006, Dumitrescu2011, Holtz-Eakin88, Juodis2021}. Further, \cite{Arkhangelsky2021, Arkhangelsky2021second, Athey2019} studied the potential outcome of treatments in panel data regimes.
	See also \cite{Hsiao} for general analysis on panel data. \\
	In the panel data setting, multiple variables are observed for the same members of a so-called {\em panel} across time.
	Concretely, suppose we are given tuples of time series $(X_i, Y_i)$, $i = 1, \dots, N$ for some $N \in \mathbb{N}$ where we succinctly refer to $X_i := \{X_{i,t}\}_{t=1}^T$.
	Moreover, assume that only one direction of causal influence is possible.
	That is, there exists $H \subseteq \{1, \dots , N\} $ such that $X_i$ causally influences $Y_i$ for all $i \in H$ but there exist no $j \in \{ 1, \dots , N\}$ such that $Y_j$ causally influences $X_j$. 
	In this setting, we refer to the set of tuples $\{(X_i, Y_i), i = 1, \dots, N\}$ as $(X, Y)$ and we say that $X$ causally influences $Y$. For testing the opposite direction, we can simply interchange the role of $X$ and $Y$.
	A simple example could be whether there is a causal link between confirmed cases and deaths resulting from COVID-19 infection in different countries. 
	Here, we are interested in the general relation of the tuple (COVID-19 confirmed cases, COVID-19 deaths) across different members of the panel (\ie, countries), see Section \ref{sec: covid experiment}. 
	
	The paper is organized as follows. In Section \ref{sec:related work} we give a concise exploration of the different approaches to Granger causality on panel data. Section \ref{sec:testing granger non-causality} is dedicated to the problem set up and proposed method, where we also provide the main theorem of the paper. 
	In Section \ref{sec: experiments}, we present experimental results for synthetic datasets. Finally, in Section \ref{sec: covid experiment}, we compare the results of state-of-the-art causal discovery algorithms on panel data against our approach on COVID-19 data about confirmed cases and deaths. 	
	\section{RELATED WORK} \label{sec:related work}
	In this section, we explore the existing literature related to our work. 
	First, we summarize different existing approaches to Granger causality on panel data. 
	We then introduce the Dumitrescu-Hurlin test, the most popular existing Granger causality test for panel data.
	Finally, we summarize an existing approach to aggregating p-values for high-dimensional regression, which provides a core component our method builds upon.
	\subsection{Hypothesis Testing of Granger Causality on Panel Data}
	
	Existing literature distinguishes between four types of hypothesis tests for Granger-causality on panel data, {\em cf.} \cite{Dumitrescu2011}. 
	Below, we formulate these tests for the causal influence from $X$ to $Y$ in terms of their null hypotheses.
	%The test from $Y$ to $X$ can be formulated analogously. 
	\\
	
	\textbf{Homogeneous Non-Causality (HNC)}: The null hypothesis of HNC test is that there is no causal relation between the variables for any individual, \ie, for all $i$ it holds that $X_i$ does not Granger-cause $Y_i$.
	\\
	\textbf{Homogeneous Causality (HC)}: Under the null, $X_i$ causes $Y_i$ for all $i$. 
	Further, it is assumed that the dynamics of $X_i$ and $Y_i$ do not change for different $i$. 
	In particular, it is assumed that the regression parameters from $X_i$ to $Y_i$ given the past of $Y_i$ are identical for all individuals. 
	\\
	\textbf{Heterogeneous Causality (HEC)}: The null is the same as for HC except that the test does not assume that the dynamics of $X_i$ and $Y_i$ remain the same across the panel. 
	\\
	\textbf{Heterogeneous Non-Causality (HENC)}: The null hypothesis is that there exists a subgroup of individuals for which $X_i$ does Granger-cause $Y_i$ and, hence, there exists at least 1 and at most $N-1$ individuals for which $X_i$ does not Granger cause $Y_i$.
	\\
	All approaches described above are based on the following assumption of the dynamics of $X$ and $Y$:
	% linear models. %, which we will carry over. 
	%Additionally, there is another important assumption: the \textit{innovation} processes of all individuals are assumed to be independent of each other.
	%In other words, the dynamics of $X$ and $Y$ follow the pattern:
	\begin{align}
		\begin{split} \label{eq::generation process}
			X_{i,t}& =  \delta_{i, 0} +  \sum_{p=1}^P \delta_{i, p} X_{i, t-p} + \eta_{i, t}, \\
			Y_{i,t} &= \theta_{i, 0} + \sum_{p=1}^P \theta_{i, p} Y_{i, t-p} + \sum_{p=1}^P \beta_{i, p} X_{i, t-p} + \epsilon_{i, t}, 
		\end{split}
	\end{align} 
	where $P$ is the time lag order, and $\delta,$ $\eta$ and $\beta$ are coefficients. The concrete assumptions about the coefficient vectors $\delta,$ $\eta$ and $\beta$ depend on the considered hypotheses HNC, HC, HEC or HENC, but in all approaches it is assumed that the innovation processes $(\{\eta_{i, t}\}, \{\epsilon_{i, t}\})_{i, t}$ are mutually independent.
	\subsection{Dumitrescu--Hurlin Test} \label{sec: DH-test}
	The most widely used test of Granger-causality in panel data is the DH test, developed by \cite{Dumitrescu2011}. 
	The test considers the HNC null hypothesis, where no Granger-causal relationships are assumed to exist for any member $i$ of the panel.
	The DH test is based on an aggregated Wald statistic of individual Granger causality tests.
	
	There exist asymptotic and semi-asymptotic characterizations of the DH test. More precisely, as stated in the introduction, we denote by $(X,Y)$ the panels under consideration. The DH test considers the generation process \eqref{eq::generation process}  where it is allowed that coefficients vary across individuals, while being time invariant.
	Since the DH test does not assume homogeneous coefficients across individuals, the null and alternative hypotheses read as follows:
	\begin{align}
		\begin{split} \label{DH null}
			H_0&: (\beta_{i,1}, \dots , \beta_{i, p} )  = (0, \dots  ,0) ~~~  \text{for all } i \\
			H_1&: \exists ~ i ~\text{ s.t. } (\beta_{i,1}, \dots , \beta_{i, p} )  \neq (0, \dots  ,0).
		\end{split}
	\end{align}
	To test this hypothesis, \cite{Dumitrescu2011} first construct a Wald statistic $W_{i, T}$ for each individual. The \textit{average Wald statistic $W^\textit{Hnc}_{N,T}$} is then given by
	\begin{align*}
		W^\textit{Hnc}_{N,T} = \frac{1}{N} \sum_{i=1}^N W_{i, T}.
	\end{align*}
	Under suitable conditions and, in particular, using the independence between individuals (since $(\{\eta_{i, t}\}_{i, t}, \{\epsilon_{i, t}\}_{i, t})$ are mutually independent ), \cite{Dumitrescu2011} Theorem 2 shows that the normalized average Wald statistic converges in distribution to the normal distribution with expectation 0 and variance 1. We include further details in Appendix \ref{sec: DH further explanation}.
	
	The proof relies on the central limit theorem and hence, the independence between innovation processes is required (see Assumption 2 of \cite{Dumitrescu2011}). 
	This implies that dependencies across panel members are assumed to not exist which is a rather restrictive assumption, since it prohibits any interaction among the individual panel members.
	In practice, however, it is often the case that such interactions exist.
	In reference to our example above, the confirmed COVID-19 cases of different countries may be causally linked, \eg, because infected people might travel to different countries and infect locals.
	
	In Section \ref{sec:testing granger non-causality}, we will take a different approach to develop a hypothesis test where interactions between the innovation processes are taken into account. This approach will rely on a p-value aggregation idea for high dimensional regression.
	Although that approach is not connected to panel data, the idea can be used in panel data settings to obtain a p-value that controls the type I error by the chosen significance level even if dependencies between individuals (\ie, between the innovation processes) exist. 
	
	Dumitrescu and Hurlin  also suggest an alternative approach in \cite{Dumitrescu2011} Section 6.2 that takes cross-sectional dependencies into account, which we call \textit{DH block bootstrap test} in the following. As stated by the authors, the suggested algorithm results in very high computational costs; therefore new panel non-causality tests should be developed to account for cross section dependencies. For completeness, we include an explanation of the procedure in Appendix \ref{sec: dh block bootstrap} and also compare our approach to the DH block bootstrap test in the experiments.
	\begin{comment}
	\begin{enumerate}
	\item Define the model to test the null hypotheses. In virtue of \eqref{eq:gen-proc} we define the model to test whether $X$ causes $Y$ by $Y_{i,t} = \theta_{i, 0} + \sum_{p=1}^P \theta_{i, p} Y_{i, t-p} + \sum_{p=1}^P \beta_{i, p} X_{i, t-p}$. 
	\item 
	Estimate the model and calculate the corresponding test statistic $Z_{N, T}^{Hnc}$ for each member of the panel. 
	\end{enumerate}
	
	\end{comment}

	\subsection{A P-Value Aggregation Method From High - Dimensional Regression} \label{sec::p-values for high-dimensional regression}
	In this section, we describe the key concepts that we use for our test procedure, which is based on the work of \cite{Meinhausen2009}, see also \cite{Dezeure2015}. As mentioned in the previous section, this concept is not connected to panel data. Instead, it considers multiple bootstrap runs in regression, generates p-values for every bootstrap sample and aggregates them such that the corresponding test statistic (of the considered test) bounds the type I error. 
	
	More precisely, \cite{Meinhausen2009} consider the following setup. Let $\boldsymbol Z$ be an $n$-dimensional response vector and $\boldsymbol W$ a $n \times k$ dimensional design matrix such that
	\begin{align*}
		\boldsymbol Z = \boldsymbol W \boldsymbol b + \boldsymbol \tau,
	\end{align*}  
	with $\boldsymbol \tau$ being an i.i.d.~$n$-dimensional random vector with $\tau_i \sim \mathcal  N(0, \sigma^2)$ for some $\sigma^2 >0$ and $ \boldsymbol b \in \mathbb{R}^k$. \cite{Meinhausen2009} consider the problem of finding all $j$ such that $b_j > 0 $. In terms of statistical significance, \cite{Meinhausen2009} aim to assign p-values for the null hypotheses
	\begin{align*}
		H_{0, j}: b_j = 0,
	\end{align*}
	where they observe $n$ samples $(z_i, w_i)$, $i = 1,\dots, n$. Further, the paper assumes a high dimensional setting, \ie, $k \gg n$, where statistical inference is challenging. In order to alleviate this issue, \cite{Wasserman2008} proposes to split the data into two parts. The first part is used for feature selection where important variables are kept with high probability. The second half of the data is used to assign p-values to the kept features by using classical least squares estimation (the p-values for the dropped features can, for instance, be set to 1). Under some weak conditions, this procedure results in an approximately correct p-value. 
	%controls the family-wise error rate (FWER) approximately by the chosen significance level. 
	
	However, \cite{Meinhausen2009} argue that this procedure relies on an arbitrary split of the $n$ samples and hence, the results can vary significantly making the test hard to reproduce.  
	This problem can be solved by a multi-splitting approach in which the procedure of \cite{Wasserman2008} is repeated $m$ times. 
	To understand the key idea of the multi-split approach of  \cite{Meinhausen2009}, we ignore the fact that we want to find all regressors for which $b_j > 0$, but rather focus on a fixed $j$ for which we want to test $H_{0, j}$.	
	In the multi-splitting approach, $m$ p-values $P_j^{(1)}, \dots P_j^{(m)}$ are generated, where each $P^{(l)}_j$, $l=1, \dots , m$ corresponding to each of the $m$ splits, is an approximately correct p-value for the test $H_{0,j}$. 
	In the next step, these p-values are aggregated to a single, approximately correct, p-value $P_j$. In contrast to the single split method of \cite{Wasserman2008}, the generated p-value of \cite{Meinhausen2009} is stable and as a result, makes the experiment reproducible. To aggregate the $m$ p-values  $P_j^{(1)},\dots, P_j^{(m)}$, \cite{Meinhausen2009} proposes the following procedure: \\
	For arbitrary $\gamma \in (0,1)$, define 
	\begin{align}
		\begin{split}
			Q_j(\gamma) &:= \min \Big\{1, \\ &\text{emp. } \gamma \text{-quantile} (\{P_j^{(l)} / \gamma,\; l=1,\dots , m\}) \Big\}. \label{eq: quantile meinhausen}
		\end{split}
	\end{align}
	Then, $P_j := Q_j(\gamma)$ is an asymptotically correct p-value, \ie, using $Q_j(\gamma)$, the type one error of the test $H_{0,j}$ is approximately bounded by the chosen significance level $\alpha \in (0,1)$ for each $\gamma \in (0,1)$, see \cite{Meinhausen2009} Theorem 3.1. 
	\\
	
	Above, we focused on the case where we are interested in a single predictor for simplicity. However, this approach can easily be generalized to the case where we want to test all predictors.~Corresponding FWER/FDR control procedures are also explained in \cite{Meinhausen2009}.
	
	\section{TESTING FOR GRANGER NON-CAUSALITY ON PANEL DATA} \label{sec:testing granger non-causality}
	\subsection{Problem Setup}
	In this paper, we consider a test for the HNC hypothesis, introduced in the previous section. 
	That is, we assume that under the null $X_i$ does not Granger-cause $Y_i$ for any $i$.
	As in the previous section, we consider the bivariate case for ease of notation but this can be easily generalized. We consider the same generation process as in \eqref{eq::generation process}:
	\begin{align}
		\begin{split} \label{eq:gen-proc}
			X_{i,t}& =  \delta_{i, 0} +  \sum_{p=1}^P \delta_{i, p} X_{i, t-p} + \eta_{i, t}, \\
			Y_{i,t} &= \theta_{i, 0} + \sum_{p=1}^P \theta_{i, p} Y_{i, t-p} + \sum_{p=1}^P \beta_{i, p} X_{i, t-p} + \epsilon_{i, t}, 
		\end{split}
	\end{align} 
	for $i = 1, \dots , N$ denoting the cross-sectional units (\ie, individuals or panel members) and $t = 1, \dots , T$ the timestamps. We also assume that the stochastic processes $Y_i$, $X_i$ are scalar. The parameters $\theta_i$ denote the heterogeneous, autoregressive coefficients of the processes. The stochastic processes $\epsilon_{i, t}$ and $\epsilon_{j, t}$ as well as $\eta_{i, t}$ and $\eta_{j, t}$  are not assumed to be independent in $i \neq j$, \ie, cross-sectional dependencies might exist. Note that in contrast to the DH-test, we do not rely on an independence assumption across panel members.
	Figure \ref{fig: example structures} illustrates different types of dependencies between members of the panel and shows how this dependence is reflected in the dependence between the innovation processes.
	\begin{figure} 
		\begin{center}
			\includegraphics[scale=0.15]{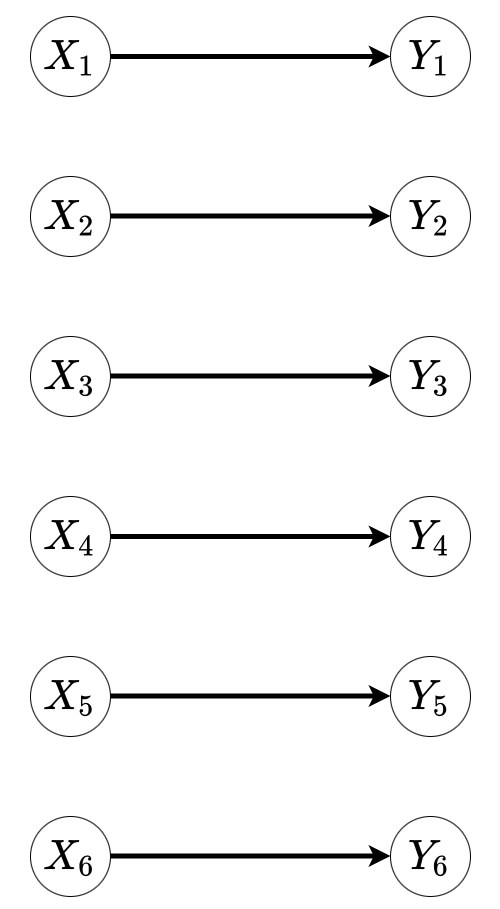}
			\includegraphics[scale=0.15]{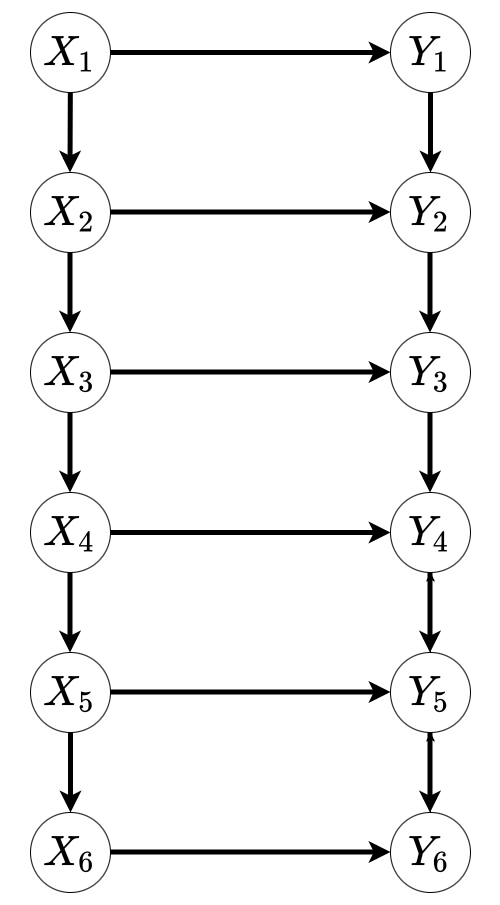} 
			\includegraphics[scale=0.15]{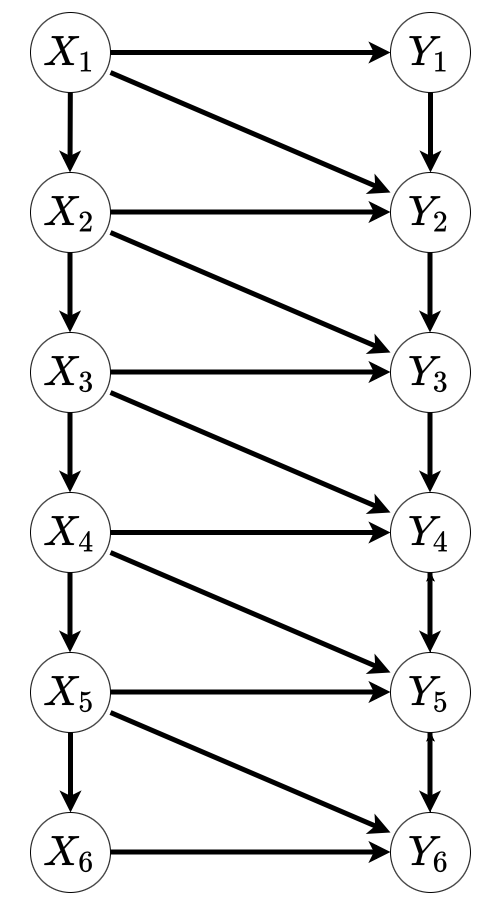}
		\end{center}
		\caption{\textit{Example of the summary graph (omitting self-cycles) of panel data where $X$ causally influences $Y$. \\
				\textbf{Left:} Panel data where no cross-sectional dependencies exist, \ie, the innovation processes are independent, namely $\epsilon_i$, $\eta_i$, $\epsilon_j$, $\eta_j$ are mutually independent for $j \neq i$. This structure satisfies the assumptions about the dependence between members for the DH-test and our  QPPA approach (see Section~\ref{sec:panel_boot}). \\
				\textbf{Middle:} Panel data with cross-sectional dependencies in the sense that $\epsilon_i$ and $\epsilon_j$ as well as $\eta_i$ and $\eta_j$ are dependent but $\epsilon_i$ and $\eta_j$ are independent for $j \neq i$. This structure \textbf{does not} satisfy the assumptions about the panel structure for the DH-test as cross-sectional dependencies exist. However, this structure \textbf{does} satisfy the assumptions about the panel structure for the QPPA approach. Further, for QPPA, it is not required that the cross-sectional dependencies are acyclic, there might be cycles in the causal dependencies between $\{X_i: i \in N\}$ as well as $\{Y_i: i \in N\}$ and it would still satisfy the assumptions of the QPPA approach.   \\
				\textbf{Right:} Panel data where cross-sectional dependencies exist where other panel members are also confounders of $X_i$ and $Y_i$. For the innovation processes this means that $\epsilon_i$ and $\epsilon_j$ as well as $\eta_i$ and $\eta_j$ are dependent \textbf{and} $\eta_j$ and $\epsilon_i$ are dependent for all $i, j$ as well. This structure \textbf{does not} satisfy the assumptions about the panel structure for the DH-test and our QPPA approach, seeing that $X_{i-1}$ would be a confounder of $X_i$ and $Y_i$. \\
				Note that this figure does not show any relations between the innovation processes in time. In general, for all settings we assume that the innovation processes are independent in time, \ie, for all  $t \neq s \neq  h \neq l$, $\epsilon_{i, t}$, $\eta_{i, s}$, $\epsilon_{j, h}$, $\eta_{j, l}$ are mutually independent (for all $i,j$).} 
		}  \label{fig: example structures}
	\end{figure}	
	Our test relies on the same hypotheses test as the DH-test, namely \eqref{DH null}. However, our algorithm is specifically designed for settings where the causal relation is the same for all individuals while the generation processes can be different, \ie, if $\beta_{i} = (\beta_{i,1}, \dots , \beta_{i, P}) \neq (0, \dots , 0)$ for one $i$, then this holds for all $i$, but not necessarily $\beta_{i} = \beta_{j}$ for $i \neq j$. If $\beta_{i} = (\beta_{i,1}, \dots , \beta_{i, P}) = (0, \dots , 0)$ for one $i$, then this holds for all $i$. The corresponding hypothesis test reads
	\begin{align}
		\begin{split}
			H_0:~ &\text{X does not Granger-cause Y, \ie,}  \\ &\beta_{i} = (\beta_{i,1}, \dots , \beta_{i, P}) = (0, \dots , 0) ~~~ \text{for all } i \label{eq:H0}
		\end{split}
		\\[10pt]
		\begin{split}
			H_1: ~&\text{X does Granger-cause Y, \ie,} \\ & \beta_{i} = (\beta_{i,1}, \dots , \beta_{i, P}) \neq (0, \dots , 0) ~~~ \text{for all } i. \label{eq:H1}
		\end{split}
	\end{align}
	Note that $X$ and $Y$ can be exchanged in this test and hence, we can test whether $X$ Granger-causes $Y$ or $Y$ Granger-causes $X$. 
	
	Compared to the DH-test described in Section \ref{sec: DH-test}, the hypothesis test \eqref{eq:H0} vs \eqref{eq:H1} is more restrictive, since we assume that all members of the panel share the same causes, whereas the DH-test does not rely on such a strong assumption. 
	That is, in the generation process of the DH-test under the alternative there exists \textit{at least} one member of the panel where $\beta_i \neq 0$.
	In contrast, under \eqref{eq:H1} it holds that for all $i$, $\beta_i \neq 0$. 
	However, this assumption is realistic in many cases, as, for instance, in the setting of our experiments on COVID-19 data. 
	For more details and explanations, see Section \ref{sec: covid experiment}. 
	Note, however, that the test procedure we will introduce can still guarantee type-I error control under the test scenario \eqref{DH null}.
	%, our test has higher power for testing \eqref{eq:H0} vs \eqref{eq:H1}.
	We show in Appendix \ref{sec: further exp synth} that our test is robust when causal connections are sporadically missing for some individuals, and that the difference between the two tests is negligible in practice when our assumption of uniformity in the existence of causal relations nearly holds.
	%In the following, let $S$ be the set of active edges and $S^c$ its complement, where we exclude self-cycles, \ie, if $X$ Granger-causes $Y$ but not vice versa, then $S$ is given as the edge from $X$ to $Y$ and $S^c$ by the edge from $Y$ to $X$. \\\\
	%To derive at a test statistic, XYZ proposed to generate a Wald-statistic $W_i$ for every individual and derived the corresponding aggregated Wald-statistic $W = \frac{1}{N} \sum W_i$ which is approximately normally distributed as long as $\epsilon_{i}$ are independent for all $i$. 
	%However, as mentioned above, this assumption is rather restrictive. 
	%Yet, the asymptotic properties of the Wald statistic do not hold when this assumption is relaxed. 
	%However, this assumptions seems rather restrictive; independence across pannels might often not hold and hence the distribution of the aggregated Wald statstic $W$ cannot be approximated by the normal distribution. 
	%On the other hand, XYZ showed that independence is not necessary to hold for $W$ to converge to a normally distributed random variable as long as the panel satisfies a mild dependence property. This, however, is also prohibits flexible modeling since an explicit dependence structure has to hold in order to recover the asymptotic distribution of $W$. \\\\
	
	\subsection{Quantile p-value Panel Adjustment (QPPA)} \label{sec:panel_boot}
	
	In this section, we describe the procedure to test \eqref{eq:H0} vs \eqref{eq:H1} using the idea described in Section \ref{sec::p-values for high-dimensional regression}. 
	Instead of using an aggregated Wald statistic (see Section \ref{sec: DH-test}), we propose a procedure that is inspired by \cite{Meinhausen2009} (Section \ref{sec::p-values for high-dimensional regression}). That is, \cite{Meinhausen2009} uses the aggregation method \eqref{eq: quantile meinhausen} to aggregate p-values of different bootstrap samples.
	Translating this idea into our setting, we can aggregate the p-values of the panel members, hence, we treat individual panel members in the same way \cite{Meinhausen2009} treats different bootstrap runs.
	The aggregation will then control the type-I error asymptotically by the chosen significance level. 
	Since we calculate a Granger non-causality p-value for every member, we further need the following technical assumptions that we carry over from \cite{Dumitrescu2011}, see also \citep{Granger69}.
	\paragraph{Time i.i.d. residuals:} For each fixed $i \in \{1, \dots , N\}$ $\epsilon_{i, t}$ are independent for all $t = 1, \dots , T$ and normally distributed with $\mathbb{E}(\epsilon_{i, t}) = 0$ and $\mathbb{E} (\epsilon_{i, t}^2) < \infty$, where  $\mathbb{E} (\epsilon_{i, t}^2)$ is assumed to be constant in $t$. 	
	\paragraph{Covariance stationarity:} For all $i$ and $t$, it holds that $X_{i, t}$ and $Y_{i, t}$ have finite variance and $\mathbb{E}(X_{i,t}X_{j,t+h})$, $\mathbb{E}(Y_{i,t}Y_{j,t+h})$, $\mathbb{E}(Y_{i,t}X_{j,t+h})$, $\mathbb{E}(Y_{i,t})$ and $\mathbb{E}(X_{i,t})$ do not depend on $t$.
	
	Under these assumptions, we can introduce our test procedure, called {\bf Quantile p-value Panel Adjustment} (QPPA) in two steps.
	
	\paragraph{\bf Step 1: Compute a p-value for every member of the panel.} \label{step 1 qppa}
	The first step is the same as for the DH-test: We apply Granger Non-causality to each individual panel member, where we use a Wald-statistic to test for the presence of Granger causality. 
	Corresponding to these Wald statistics, we obtain an asymptotically correct p-value $p_{X \rightarrow Y}^{(i)}$ for each panel member.
	For instance, an $f$-statistic (which belongs to the family of Wald-statistics) can be used to calculate the corresponding $p$-values.
	Note that the construction of the $f$-statistic requires the assumption of time i.i.d. residuals as well as covariance stationarity. 
	Further, in order for the p-values to be (asymptotically) correct, $\eta_{j, t}$ and $\epsilon_{i, s}$ need to be independent for all $t$, $s$, $j$ and $i$. This requirement is necessary in general and does not impose additional restrictions to the generation process compared to the existing literature for Granger causality (see Figure \ref{fig: example structures} right for an example where this assumption is violated).
	\paragraph{Step 2: Aggregate p-values.}
	Similar to the procedure in Section \ref{sec::p-values for high-dimensional regression}, we aggregate the computed p-values as follows.
	For $\gamma \in (0,1)$, we define
	\begin{align}
		\begin{split}
			Q_{X \rightarrow Y}(\gamma) &:=\min \bigg\{1,\\  &\text{emp. } \gamma \text{-quantile} \left\{p_{X \rightarrow Y}^{(i)} / \gamma; ~ i=1, \dots , N \right\}   \bigg\}, \label{eq: def Q}
		\end{split}	
	\end{align}
	
	We can now formulate the main theorem of our work, the proof of which can be found in Appendix \ref{a:proofs}.
	\begin{thm} \label{thm: main 1}
		Assume the generation procedure \eqref{eq:gen-proc} and let $\alpha, \gamma \in (0,1)$. Then, $Q_{X \rightarrow Y}(\gamma)$ is an asymptotically correct p-value, \ie,
		\begin{align*}
			\limsup_{T \rightarrow \infty} \mathbb{P}\left( Q_{X \rightarrow Y}(\gamma) \le \alpha    \right) \le \alpha,
		\end{align*}
		where $T$ denotes the number of timestamps. 
	\end{thm}
	With this theorem, we obtain type-I error control for arbitrary $\gamma \in (0,1)$. 
	%	As suggested by \cite{Meinhausen2009}, a practical choice of $\gamma$ is $0.5$. 	
	Note that the proof does not require any restrictions on the dependence between p-values. In particular, the individuals do not need to be independent, \ie, restrictions on the relation between $\epsilon_{i, t}$ and $\epsilon_{j, t}$ as well as $\eta_{i, t}$ and $\eta_{j, t}$ are not required.
	
	Although Theorem~\ref{thm: main 1} holds for arbitrary $\gamma$, choosing the right $\gamma$ could be difficult.
	While \cite{Meinhausen2009} recommend $\gamma = 0.5$ as a practical choice, they also construct another p-value which only requires to specify a lower bound for $\gamma$.
	This approach can also be translated into the panel data setting. 
	We include the corresponding procedure in Appendix \ref{sec: qppa thm2}.
	
	Note that, although our approach is based on the procedure of \cite{Meinhausen2009} which deals with high dimensional statistics, we do not consider a high dimensional setup here. The connection between  \cite{Meinhausen2009}  and our approach is the problem that a single p-value results in an unstable procedure which is hence hard to reproduce.
	
	Throughout this paper, we only consider the bivariate case for simplicity but this can easily be generalized and corresponding  FWER/FDR procedures based on \cite{Meinhausen2009} can be constructed. Moreover, since $Q_{X \rightarrow Y}(\gamma)$ is an asymptotically correct p-value, also classical FWER/FDR control procedures (e.g. Bonferoni or Benjamini-Hochberg) can be applied in the multivariate case.
	
	Before moving to numerical experiments, we briefly discuss some implications of our assumptions and potential limitations of our approach.
	Both covariance stationarity and residuals that are i.i.d. in time are standard assumptions in Granger causality tests.
	Although covariance stationarity is a strong assumption, different data preprocessing procedures could be applied to obtain stationarity in practice \citep{Hyndman2018}, also see Appendix \ref{non-stationary}. 
	For example, we employ such methods to make the COVID-19 time series stationary in our study in Section \ref{sec: covid experiment}.
	Similarly, assuming residuals are i.i.d.~or Gaussian are strong assumptions. 
	Our experiments on COVID-19 data and further experiments in Appendix \ref{sec: further exp synth} examine the robustness of our algorithm against the violation of this assumption. 
	Finally, the generation process \eqref{eq:gen-proc} has two further implications: assuming $\epsilon$ and $\eta$ are independent implicitly assumes causal sufficiency (\ie, no hidden confounders), and that there are no {\em instantaneous effects} between $X_t$ and $Y_t$.
	Again, both assumptions are typical in Granger causality analysis.
	We discuss how to decrease the false discovery rate in the presence of confounding, subject to causal faithfulness, in Appendix \ref{confounding} and the robustness of our approach under instantaneous effects in Appendix \ref{instantenous effects}.

	\begin{comment}
	If we are interested in testing both directions, \ie, $X$ causally influences $Y$ and $Y$ causally influences $X$, we can simply adjust the constructed p-values for multiplicity using, for instance, Bonferroni correction, and hence obtain p-values that control the FWER.
	\end{comment}
	\section{EXPERIMENTS WITH SYNTHETIC DATA} \label{sec: experiments}
	In this section, we present experimental results on synthetic data. We compare the DH test, DH block bootstrap method and our approach (QPPA).
	\begin{comment}
	Here, we intentionally exclude the generation process according to Figure \ref{fig: example structures} right. This is because $X_{i-1}$ acts as a hidden common cause between $X_i$, $Y_i$, if we consider Granger causality for the individual $X_i$, $Y_i$ and do not regress on the other individuals (which is the case in the panel data approach). Hence, such a test will not give additional insights since both approaches will suffer from the hidden common cause problem.
	\end{comment}
	To compare these methods, we use the existing \texttt{xtg-cause} package developed by \cite{Lopez2017}, which includes the DH-test and the DH block bootstrap test. 
	For the QPPA approach presented in Section~\ref{sec:testing granger non-causality}, we use our own implementation where we also rely on the Granger causality test implementation in \texttt{statsmodels} \citep{seabold2010statsmodels} to compute f-statistics and corresponding p-values as explained in Section \ref{step 1 qppa}.
	We compare these approaches in two scenarios, respectively without and with cross-sectional dependencies, detailed below.
	
	\paragraph{Experiment 1:} For the first experiment, we do not insert cross-sectional dependencies. 
	The generation process, an autoregressive process of order 1, is specified as follows:
	\begin{align*}
		X_{i, t} &=  \delta_{i, 1} X_{i, t-1} + \eta_{i, t} \\
		Y_{i, t} &=  \theta_{i, 1} Y_{i, t-1} + \beta_{i, 1} X_{i, t-1} + \epsilon_{i, t},
	\end{align*} 
	where the innovation processes are i.i.d.~Gaussian random variables
	%\begin{align*}
	%	\eta_{i, t} =  \zeta, ~~~~~~~	\epsilon_{i, t} = \xi 
	%\end{align*}
	with $\eta_{i, t}, \epsilon_{i, t} \sim N(0,0.1)$ and we draw the parameters from a uniform distribution: $\delta_{i, 1},  \theta_{i, 1} \sim \textit{Unif}(0.2, 0.8)$.
	
	\paragraph{Experiment 2:} For this experiment we insert cross-sectional dependencies for $X$ and $Y$:
	\begin{align*}
		X_{i, t} &=  \delta_{i, 1} X_{i, t-1} + \zeta_{i, t} \\
		Y_{i, t} &=  \theta_{i, 1} Y_{i, t-1} + \beta_{i, 1} X_{i, t-1} + \xi_{i, t},
	\end{align*} 
	where $\boldsymbol{\zeta}_{t} := (\zeta_{1, t}, \dots \zeta_{N, t}) \sim N(\boldsymbol{0},\Sigma)$, $\boldsymbol{\xi}_{t} := (\xi_{1, t}, \dots \xi_{N, t}) \sim N(\boldsymbol{0},\tilde{\Sigma})$, $\Sigma = A^T A, \tilde{\Sigma} = \tilde{A}^T \tilde{A}$ and $A, \tilde{A}$ are random vectors where each entry is sampled from $Unif(0.5,1.5)$. Finally, $\delta_{i, 1}, \theta_{i, 1} \sim \textit{Unif}(0.2, 0.8)$ as in Experiment 1.
	The cross-sectional dependency is embedded in the multivariate normal distribution of the noise terms $\boldsymbol{\zeta}_t$ and $\boldsymbol{\xi}_t$. 
	Note however that there is no dependence between $\boldsymbol{\zeta}_t$ and $\boldsymbol{\xi}_t$ and also no dependence of the noise terms in time.
	
	In both experiments we either sample $\beta_{i, 1}$ from $\textit{Unif}(0.2, 0.8)$ if the null should be rejected and set $\beta_{i, 1} = 0$ if the null should be accepted. 
	
	We report power and false discovery rates (FDR) for DH-test, DH-test with block bootstrap (\texttt{DH-test-bb}), and QPPA in Tables~\ref{fig:empirical results} and \ref{fig:empirical results exp 2} for experiments 1 and 2 respectively. 
	Each number reported is an average of 100 experiments.
	For QPPA we use $\gamma = 0.5$ in \eqref{eq: def Q} and for the DH-test and DH-test with block bootstrap we use the statistic $\tilde{Z}_N^{HNC}$, see Appendices \ref{sec: dh block bootstrap} and \ref{sec: DH further explanation}.
	
	Table \ref{fig:empirical results} shows that QPPA performs equally well as DH and DH-bb in the setting without cross-sectional dependencies with a sufficiently large history  (T>10).
	The results in Table \ref{fig:empirical results exp 2} exhibit that if cross-sectional dependencies exist, the FDR of the DH-test increases significantly above the significance level 0.05 even in the large sample regime. Also DH-test-bb shows higher FDR than 0.05 even in the large sample regime. This is not the case for QPPA which remains robust against this type of dependency. Moreover, the power of both tests is 1 in the large sample regime. 
	\begin{table} [ht]
			\caption{\textit{Empirical results for Experiment 1 (no cross-sectional dependencies)}} \label{fig:empirical results}
		\centering
		\scalebox{0.7}{
			\begin{tabular}{llllllll}
				\toprule
				&      & \multicolumn{2}{l}{QPPA} & \multicolumn{2}{l}{DH-test} & \multicolumn{2}{l}{DH-test-bb} \\
				&      & Power &      FDR &   Power &      FDR &              Power &      FDR \\
				\midrule
				T=10 & N=1 &  0.15 &  0.211 &    0.35 &  0.186 &               0.33 &  0.154 \\
				& N=10 &   0.000 &        0.000 &    0.93 &  0.212 &               0.74 &  0.119 \\
				& N=30 &   0.000 &        0.000 &     1.0 &  0.359 &               0.97 &  0.110 \\
				T=50 & N=1 &  0.77 &   0.038 &    0.83 &  0.117 &               0.83 &  0.126 \\
				& N=10 &  0.98 &      0.000 &     1.0 &  0.074 &                1.0 &  0.074 \\
				& N=30 &   1.0 &      0.000 &     1.0 &  0.048 &                1.0 &  0.083 \\
				T=100 & N=1 &  0.91 &  0.022 &    0.94 &  0.078 &               0.97 &  0.093 \\
				& N=10 &   1.0 &      0.000 &     1.0 &  0.065 &                1.0 &  0.074 \\
				& N=30 &   1.0 &      0.000 &     1.0 &  0.074 &                1.0 &  0.107 \\
				\bottomrule
			\end{tabular}
		}
	\end{table}
	\begin{table} [ht]
		\caption{\textit{Empirical results for Experiment 2 (with cross-sectional dependencies)}} \label{fig:empirical results exp 2}
		\centering
		\scalebox{0.7}{
			\begin{tabular}{llllllll}
				\toprule
				&      & \multicolumn{2}{l}{QPPA} & \multicolumn{2}{l}{DH-test} & \multicolumn{2}{l}{DH-test-bb} \\
				&      & Power &      FDR &   Power &      FDR &              Power &      FDR \\
				\midrule
				T=10 & N=1 &  0.170 &     0.320 &    0.420 &  0.236 &               0.340 &  0.261 \\
				& N=10 &  0.120 &      0.000 &    0.630 &  0.344 &               0.320 &  0.220 \\
				& N=30 &  0.160 &  0.059 &    0.810 &  0.449 &               0.380 &  0.191 \\
				T=50 & N=1 &   0.800 &  0.048 &    0.850 &  0.086 &               0.850 &  0.105 \\
				& N=10 &  0.920 &      0.000 &     1.0 &  0.180 &               0.990 &  0.083 \\
				& N=30 &  0.960 &      0.000 &     1.0 &  0.408 &               0.990 &  0.075 \\
				T=100 & N=1 &  0.940 &  0.051 &    0.970 &  0.049 &               0.980 &  0.110 \\
				& N=10 &  0.970 &   0.010 &     1.0 &  0.174 &                1.0 &  0.074 \\
				& N=30 &   1.0 &      0.000 &     1.0 &   0.419 &                1.0 &  0.082 \\
				\bottomrule
			\end{tabular}
			
		}
	\end{table}
	For low sample sizes, QPPA is rather conservative (for the choice $\gamma = 0.5$), hence it has low power and low FDR for $T=10$. 
	The opposite holds for the DH-test. 
	Note that for all approaches we set $\alpha = 0.05$. However, in  the experiments, we observe lower FDR than the chosen significance level for QPPA. This is due to the fact that the p-values of QPPA are not uniformly distributed but $\alpha$ is only an upper bound for the false discovery rate, see Theorem \ref{thm: main 1}.  
	While we heuristically set $\gamma = 0.5$ for results reported in Figures~\ref{fig:empirical results} and \ref{fig:empirical results exp 2},
	we show in Figures~\ref{fig: varying gamma exp 1} and \ref{fig: varying gamma exp 2} how power and FDR vary for $\gamma = 0.01, 0.02, \dots , 0.99$.
	The Figure shows that in these experiments, the power decreases with increasing $\gamma$. However, the FDR is 0 for all $\gamma$ due to the well separation between the distribution under the null and alternative.
	\begin{figure}[ht]
		\centering
		\includegraphics[scale=0.4]{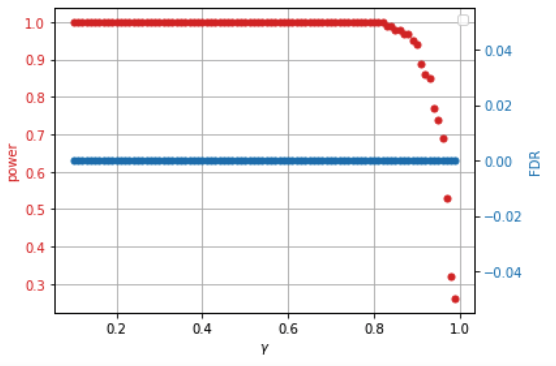}
		\caption{\textit{Empirical results in the experimental setup 1 for QPPA for $\gamma = 0.01, 0.02 \dots 0.99$ }, where we set $T=100$ and $N=30$.} \label{fig: varying gamma exp 1}
	\end{figure}
	\begin{figure}[ht]
		\centering
		\includegraphics[scale=0.4]{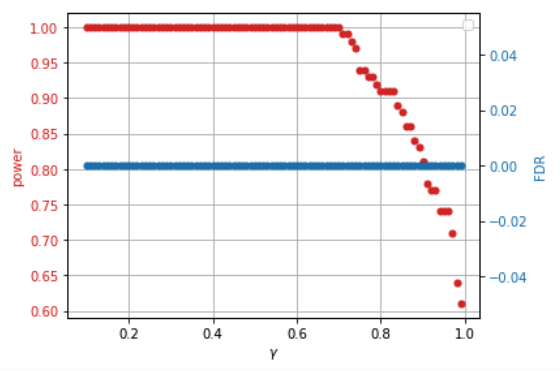}
		\caption{\textit{Empirical results in the experimental setup 2 for QPPA for $\gamma = 0.01, 0.02 \dots 0.99$ }, where we set $T=100$ and $N=30$.} \label{fig: varying gamma exp 2}
	\end{figure}
	We add further experiments in Appendix \ref{sec: further exp synth}, where we also include additional results about the behavior of QPPA for varying $\gamma$.
	\section{EXPERIMENTS WITH COVID-19 DATA} \label{sec: covid experiment}
	While synthetic data experiments confirm our test performs in line with expectations, we carry out a second set of experiments to demonstrate both the robustness of our approach on real-world problems, and how the problem we address arises ubiquitously.
	To this end, we give an illustrative study of causal discovery in time series related to COVID-19.
	Here, we regard individual countries and regions as panel members, and we analyze the causal relation between confirmed cases and deaths within each country.
	%Specifically, we want to analyze causal influences of past two weeks to the present. 
	Since increasing number of confirmed cases lead to more deaths (after several days to approximately two weeks), we regard
	\begin{align*}
		\text{confirmed cases } \longrightarrow \text{deaths}
	\end{align*} 
	as the ground truth causal relation, and analyze causal influences from the past two weeks to the present.
	
	Some remarks on the ground truth causal relation: Although one might argue that deaths can causally influence confirmed cases because more vulnerable immune systems reduce as the number of deaths increase, we expect such effects to take place over longer periods of time that should not occur in the scope of our study. \\ Further, if a huge proportion of death cases did not get tested before death, death cases can causally influence confirmed cases because died people might get tested post mortem. However, confirmed cases is a proxy for the number of infections and, additionally,
	COVID-19 deaths where people are not going to the hospital and die without confirmation but get tested post mortem
	are rather rare and should not be statistically significant.

	As one might expect, cross-sectional dependencies among countries should exist, \eg, through active cases traveling between countries (see Appendix \ref{sec: covid panel dependence} for an analysis confirming this hypothesis).
	%since the number of more vulnerable immune systems decreases with increasing death cases and hence as a result the confirmed cases decrease, this counter effect can be excluded since it takes a very long time (decades rather than weeks) until this effect has a significant impact on the pandemic. \\
	%Consequently, bidirectional or no causal relation from confirmed cases to death cases can be excluded for all members of the panel hence we are in the setting where we test \eqref{eq:H0} vs \eqref{eq:H1}.  
	%Moreover, note that we have no guarantee that our assumptions about the data generating process, such as Gaussian innovations, hold true in this real data study.
	
	We use the Johns Hopkins CSSE COVID-19 data repository\footnote{The data set can be downloaded at \url{https://github.com/CSSEGISandData/COVID-19/tree/master/csse\_covid\_19\_data/csse\_covid\_19\_time\_series}, from this repository we used the global csv files} \citep{dong2020}. 
	The data set contains confirmed cases and deaths due to COVID-19, collected globally in 280 countries or regions, where cases are recorded daily between 22nd January 2020 and 4th October 2021. 
	After applying preprocessing steps outlined in Appendix \ref{sec: covid preprocessing}, there remain 335 days of data for 225 panel members (\ie, countries or regions). 
	%We also apply cross-sectional independence tests to check whether cross-sectional dependencies exist, using the Stata package \texttt{xtcsd}. 
	%%For this, we  which tests cross section independence in panel data. 
	%The results, where we confirm cross-sectional dependencies between countries exist, are given in Appendix \ref{sec: covid panel dependence}. 
	%Thereby, in our study the null gets rejected, meaning that the alternative that cross-sectional dependencies exist, gets accepted. \\
	In Table \ref{fig: covid results} we give p-values of the DH-test, DH-test with block bootstrap and QPPA for the causal influence from confirmed cases to deaths and visa versa, where we allow a time order in the underlying models to allow influence for up to two weeks.
	\begin{table}[ht]
		\centering
		\caption{\textit{Results of the covid-19 causal discovery study. p-val QPPA relates to the p-values obtained by our QPPA approach with $\gamma=0.5$ in \eqref{eq: def Q}, p-val DH-test to the p-values obtained by the DH-test and p-val DH-test-bb to the p-values obtained by the DH-test with block bootstrap, where we use 20 breps, see Appendix \ref{sec: dh block bootstrap} and the statistics $\tilde{Z}_N^{HNC}$, see Appendix \ref{sec: DH further explanation}. $c \rightarrow d$ is the corresponding p-value to the causal link "confirmed cases causes deaths" and $d \rightarrow c$ the p-value to the causal link "deaths causes confirmed cases". }} \label{fig: covid results}
		\scalebox{0.7}{\begin{tabular}{crrllll}
				\toprule
				P (=lag order) & \multicolumn{2}{l}{p-val QPPA} & \multicolumn{2}{l}{p-val DH-test} & \multicolumn{2}{l}{p-val DH-test-bb} \\
				{}& c -> d & d -> c &        c -> d & d -> c &            c -> d & d -> c \\
				\midrule
				1 &      0.610 &  0.607 &         0.000 &  0.000 &             0.000 &  0.000 \\
				2 &      0.323 &  0.343 &         0.000 &  0.000 &             0.000 &  0.000 \\
				3 &      0.183 &  0.239 &         0.000 &  0.000 &             0.000 &  0.000 \\
				4 &      0.091 &  0.133 &         0.000 &  0.000 &             0.000 &  0.000 \\
				5 &      0.055 &  0.094 &         0.000 &  0.000 &             0.000 &  0.000 \\
				6 &      0.036 &  0.084 &         0.000 &  0.000 &             0.000 &  0.000 \\
				7 &      0.015 &  0.110 &         0.000 &  0.000 &             0.000 &  0.000 \\
				8 &      0.005 &  0.064 &         0.000 &  0.000 &             0.000 &  0.000 \\
				9 &      0.003 &  0.082 &         0.000 &  0.000 &             0.000 &  0.000 \\
				10 &      0.002 &  0.080 &         0.000 &  0.000 &             0.000 &  0.000 \\
				11 &      0.001 &  0.065 &         0.000 &  0.000 &             0.000 &  0.000 \\
				12 &      0.001 &  0.056 &         0.000 &  0.000 &             0.000 &  0.000 \\
				\bottomrule
		\end{tabular}}
	\end{table}
	%We test up to 12 time lags since according to the preprocessing steps explained in Appendix \ref{sec: covid preprocessing} we generated the second order difference to make the time series stationary, and hence the model contains information about up to 14 past days. 
	The results show that the DH-test and DH-test with block bootstrap reject the null for both directions with p-value of 0.000 for all lags, leading to a wrong conclusion (confirming that deaths also cause confirmed cases).
	In contrast, QPPA rejects the null only for the correct direction (confirmed cases $\rightarrow$ deaths) to the significance level of 5\% after a reasonable time order is given. 
	Further, it is reasonable that the null from confirmed cases to deaths gets only rejected after including multiple lags, because infection (more precisely confirmation of infection that is recorded in the data under consideration) with COVID-19 causes death with some time delay. 
	Figure \ref{fig: cov diff gamma} shows power and FDR of QPPA for $\gamma = 0.01, 0.02, \dots , 0.99$. Compared to the corresponding study in Section \ref{sec: experiments} shown in Figure \ref{fig: varying gamma exp 1} and \ref{fig: varying gamma exp 2}, we can see that the FDR and power for small $\gamma$ is large and monotonically decreasing with increasing $\gamma$. Hence, choosing $\gamma$ amounts to a trade-off between high power and low FDR. We include more results of the behavior of QPPA for varying $\gamma$ in Appendix \ref{sec: additional covid exp} where we use different numbers of countries/regions per run.
	
	We conclude from our results that failing to account for cross-sectional dependencies can easily result in wrong conclusions using baseline methods, and that QPPA not only mitigates this risk but is also robust against real-world cases where some of its assumptions about the data generating process (\eg, Gaussianity of innovations) do not hold.
	\begin{figure}[ht]
		\centering
		\includegraphics[scale=0.40]{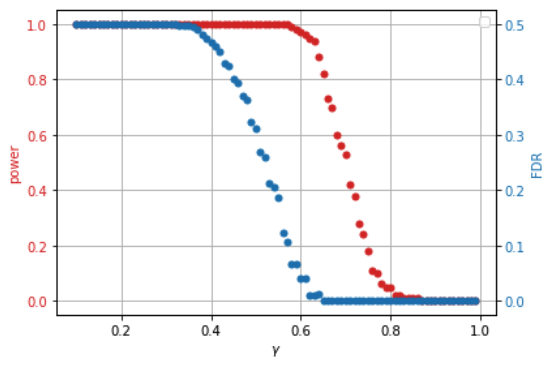}
		\caption{\textit{Empirical results for COVID-19 data about confirmed cases and deaths using QPPA with $\gamma = 0.01, 0.02, \dots, 0.99$}. To calculate power and FDR, we randomly selected 60 countries/regions out of the 225 and checked whether QPPA detects the causal relation $c \rightarrow d$ and $d \rightarrow c$ respectively to the significance level 5\% and repeat this 100 times. } \label{fig: cov diff gamma}
	\end{figure}

	In general, if one wants to recover the unknown causal relations of a panel time series dataset, we recommend to apply QPPA for varying $\gamma$. If the results are consistent for a large interval of $\gamma$, there is strong evidence to believe that the obtained causal structure for these $\gamma$ captures the relations of the underlying generative process.
	\section{CONCLUSION AND OUTLOOK}
	In this paper, we propose a new approach to causal discovery on panel data with Granger causality using a quantile based p-value adjustment approach.
	We calculate a Wald statistic for each individual and correspondingly generate an asymptotically correct aggregated p-value for the panel. 
	Compared to the most widely used causal discovery method on panel data, the DH-test, we aggregate p-values instead of the individual Wald statistics. 
	In this way, we are able to account for cross-sectional dependencies that may exist among individuals in the panel.
	
	Numerical experiments on synthetic data confirm that our approach outperforms both the DH-test and its block bootstrap variant, designed to address cross-sectional dependencies.
	Notably, in contrast to baseline methods, our approach also correctly discovers causal relationships in a real world scenario -- the cause-effect relationship between COVID-19 infections and deaths.
	Our results show that not capturing cross-sectional dependencies easily leads to incorrect conclusions, and also that our method is robust in real-world settings where some of its statistical assumptions may be violated.
	
	%	In our experiments, we compare the DH-test and our proposed approach. Further, we also compared QPPA to the variation of DH-test that takes cross-sectional dependencies into account by using a bootstrap approach (DH-test-bb). In the experiments, we found that QPPA outperforms the DH-test as well as DH-test-bb when cross-sectional dependencies exist. \\
	%	More notably, we showed in a COVID-19 study that QPPA is able to find the true causal relation between confirmed and death cases while the DH-test as well as the DH-test with block bootstrap are not.  \\
	
	Finally, we emphasize that the p-value aggregation method we employ is a general approach and does not depend on the specific generation process. 
	Therefore, this method could be applied in other applications or causal discovery methods for panel data (\eg, on non-time series data), which remains an exciting avenue for further research.
	\bibliography{ref_arxiv}
	%\bibliographystyle{plain}

%%%%%%%%%%%%%%%%%%%%%%%%%%%%%%%%%%%
%%%%%% SUPPLEMENT (OPTIONAL) %%%%%%
%%%%%%%%%%%%%%%%%%%%%%%%%%%%%%%%%%%

\clearpage
\appendix

\thispagestyle{empty}
\appendix

% For one-column format, uncomment the following:
%\onecolumn \makesupplementtitle
% For two-column format, uncomment the following:
%\twocolumn[ \makesupplementtitle ]

\section{FURTHER EXPLANATIONS OF THE DH TEST} \label{sec: DH further explanation}
Corresponding to the average Wald-static 
\begin{align*}
	W_{N, T}  := \frac{1}{N} \sum_{i=1}^{N} W_{i, T},
\end{align*}
\cite{Dumitrescu2011} show that for the normalized statistics under some regularity conditions it holds
\begin{align} 
	\begin{split}
		Z_{N, T}^{Hnc} &:= \sqrt{\frac{N}{2P}}(W_{N, T} - P) \stackrel{T, N \rightarrow \infty}{\longrightarrow} N(0,1 ) \\[5pt]
		\tilde{Z}_N^\textit{Hnc}&:= \frac{\sqrt{N}\left[ W^\textit{Hnc}_{N,T} - N ^{-1} \sum_{i=1}^N \mathbb{E}(W_{i, T}) \right]}{\sqrt{N^{-1} \sum_{i=1}^N \textit{Var}(W_{i, T})}} \stackrel{N \rightarrow \infty}{\longrightarrow} N(0,1) \label{eq::DH convergence}
	\end{split}
\end{align} 
in distribution, where $N(0,1)$ denotes the normal distribution with expectation 0 and variance 1 (where $P$ denotes the lag order, $T$ the number of timestamps and $N$ the number of individuals).
Eq. \eqref{eq::DH convergence} can now be used to test \eqref{DH null}:  If the probability of obtaining the realization of $\tilde{Z}_N^\textit{Hnc}$, $Z_{N, T}^{Hnc}$ respectively w.r.t. the standard normal distribution is low (corresponding to the chosen significance level), $H_0$ is rejected. \cite{Juodis2021} constructed a different test statistic in the same setting with the additional benefit that, in contrast to the DH-test, it accounts for "Nickell" bias which occurs if $N/T^2\rightarrow 0$ does not hold. 

\section{DH BLOCK BOOTSTRAP TEST} \label{sec: dh block bootstrap}
The DH block bootstrap procedure relies on a resampling idea, \textit{cf.} \cite{Dumitrescu2011} Section 6.2.:
\begin{enumerate}
	\item Define the model for each panel member to test Granger causality. Here, the model under consideration is
	\begin{align*}
		Y_{i,t} &= \theta_{i, 0} + \sum_{p=1}^P \theta_{i, p} Y_{i, t-p} + \sum_{p=1}^P \beta_{i, p} X_{i, t-p},
	\end{align*}
	see \eqref{eq::generation process}.
	\item
	Estimate the model and compute the corresponding test statistics $Z^{Hnc}_{N, T}, \tilde{Z}^{Hnc}_{N}$ for each panel member.
	\item
	Estimate the model under the null (no Granger causality, see \ref{DH null}), \textit{i.e.}, estimate a model 
	\begin{align*}
		Y_{i,t} &= \tilde{\theta}_{i, 0} + \sum_{p=1}^P \tilde{\theta}_{i, p} Y_{i, t-p}
	\end{align*}
	for each panel member and compute the residual vectors of size $(T, 1)$.
	\item
	Resample the residuals with replacement for each panel member with "blocks" of size 1, if we want to take dependencies across time into account, we could increase the block size.
	\item
	Construct resampled time series $\tilde{Y}_{i, t}$ under the null:
	\begin{align*}
		\tilde{Y}_{i,t} &= \tilde{\theta}_{i, 0} + \sum_{p=1}^P \tilde{\theta}_{i, p} \tilde{Y}_{i, t-p} + \tilde{\epsilon}_{i, t}
	\end{align*}
	where $(\tilde{\epsilon}_{i, t})_{t}$ denotes the resampled noise of the $i$-th panel member.
	\item
	Estimate the model defined in step 1 for the resampled time series $\{\tilde{Y}_{i, t}\}_{t}$ and construct the statistics of step 2 for this model and resampled time series.
	\item
	Repeat steps 5 and 6 a large amount of times.
	\item
	Compare the test statistics of step 2 against the test statistics obtained from steps 5 - 7.
\end{enumerate}
Since this procedure relies on resampling, we create dependencies especially in low sample regimes. Further, it relies on the generation of a new dataset in step 5 which could be not robust to violation of assumptions on real data. The DH block bootstrap approach is implemented in the \texttt{xtg-cause} library developed by \cite{Lopez2017}. Different parameters can be specified, \textit{e.g.}, the number of lags (this parameter can also be specified for the DH-test in the same library) and the number of \textit{breps} which denotes the number of repetitions of step 7.  
\section{QPPA WITH LOWER BOUND ON $\boldsymbol \gamma$} \label{sec: qppa thm2}
As explained in Section \ref{sec:panel_boot}, our aggregation procedure relies on the choice of a $\gamma \in (0,1)$.
Since a proper selection of $\gamma$ is difficult, a different aggregation method is proposed. Specify a lower bound for $\gamma$, which we denote by $\gamma_{\min} \in (0,1)$ and correspondingly define 

\begin{align}
	P_{X \rightarrow Y} := \min \left( (1- \log \gamma_{\min}) \inf_{\gamma \in (\gamma_{\min}, 1)} Q_{X \rightarrow Y} (\gamma) , ~~1   \right) \label{p-agreg}
\end{align}
for some fixed $\gamma_{\min} \in (0,1)$, where a recommended choice is $\gamma = 0.05$. Then, it holds:
\begin{thm} \label{thm: main 2}
	Assume the generation procedure \eqref{eq:gen-proc} and let $\alpha, \gamma \in (0,1)$. Then, $P_{X \rightarrow Y}(\gamma)$ is a asymptotically correct p-value, \ie,
	\begin{align*}
		\limsup_{T \rightarrow \infty} \mathbb{P}\left(  P_{X \rightarrow Y}(\gamma) \le \alpha    \right) \le \alpha,
	\end{align*}
	where $T$ denotes the number of timestamps. 
\end{thm}
Similar to Theorem \ref{thm: main 1}, the aggregation idea \eqref{p-agreg} relies on an aggregation idea from high dimension statistics (see \cite{Meinhausen2009} Theorem 3.2).
\section{PROOFS} \label{a:proofs}
\subsection{Proof of Theorem \ref{thm: main 1}}
\begin{proof}
	We follow the idea of the proof of \cite{Meinhausen2009} Theorem 3.1. Therefore, let 
	\begin{align}
		f_{X \rightarrow Y}(u) := \frac{1}{N} \sum_{i=1}^N \mathds{1}\{p_{X \rightarrow Y}^{(i)} \le u \}, ~~~u \in (0,1). \label{eq: def f}
	\end{align}
	Notice that 
	\begin{align}
		Q_{X \rightarrow Y} (\gamma) \le \alpha ~~~ \Leftrightarrow ~~~  f_{X \rightarrow Y} (\alpha \gamma) \ge \gamma  \label{eq: equiv f q}
	\end{align}
	and hence 
	\begin{align*}
		\mathbb{P} \left(Q_{X \rightarrow Y}(\gamma) \le \alpha  \right)  =   \mathbb{P}\Big( f_{X \rightarrow Y}(\alpha \gamma) \ge \gamma  \Big). 
	\end{align*}
	%\begin{align*}
	%	\mathbb{P} \left( \min_s Q_s(\gamma) \le \alpha  \right)  &\le \sum_{s \in S} \mathbb{E}\Big( \mathbbm{1}\{ Q_s(\gamma) \le \alpha \}  \Big) \\ &= \sum_{s \in S} \mathbb{E}\Big( \mathbbm{1} \left\{ f_s(\alpha \gamma) \ge \gamma) \right\} \Big) \\ &= \sum_{s \in S} \mathbb{P}\Big( f_s(\alpha \gamma) \ge \gamma  \Big). 
	%\end{align*}
	Using the Markov inequality, we have
	\begin{align*}
		\mathbb{P}\Big( f_{X \rightarrow Y}(\alpha \gamma) \ge \gamma  \Big)& \le \frac{1}{\gamma}  \mathbb{E}\big( f_{X \rightarrow Y}(\alpha \gamma)   \big) \\
		&=  \frac{1}{\gamma} \frac{1}{N} \sum_{i=1}^N  \mathbb{E} \left(   \mathds{1}\{p_{X \rightarrow Y}^{(i)} \le \alpha \gamma  \} \right). \\
		&= \frac{1}{\gamma} \frac{1}{N} \sum_{i=1}^N  \mathbb{P} \left(  p_{X \rightarrow Y}^{(i)} \le \alpha \gamma  \right)
		%& \le \frac{1}{\gamma} \frac{1}{N} \sum_{i=1}^N  \alpha \gamma,
	\end{align*}
	Since $p^{(i)}_{X \rightarrow Y}$ is by assumption asymptotically correct for each $i$, it holds 
	\begin{align*}
		\limsup_{T \rightarrow \infty}  \mathbb{P} \left(  p_{X \rightarrow Y}^{(i)} \le \alpha \gamma  \right) \le  \alpha \gamma ~~~\text{for all } i
	\end{align*}
	and hence
	\begin{align*}
		\limsup_{T \rightarrow \infty}  	\mathbb{P}\Big( f_{X \rightarrow Y}(\alpha \gamma) \ge \gamma  \Big)& \le \lim_{T \rightarrow \infty}   \frac{1}{\gamma} \frac{1}{N} \sum_{i=1}^N  \mathbb{P} \left(  p_{X \rightarrow Y}^{(i)} \le \alpha \gamma  \right) \\
		&\le \frac{1}{\gamma} \frac{1}{N} \sum_{i=1}^N  \alpha \gamma = \alpha 
	\end{align*}
	which completes the proof.
\end{proof}
\subsection{Proof of Theorem \ref{thm: main 2}}
\begin{proof}
	We follow the idea of the proof of \cite{Meinhausen2009} Theorem 3.2. Therefore, first note that for a uniformly distributed random variable $U$, it holds
	\begin{align}
		\sup_{\gamma \in (\gamma_{\min} , 1)}  \frac{\mathds{1}\{ U \le \alpha \gamma \}}{\gamma} = 
		\begin{split} \label{eq: split sup U}
			\begin{cases}
				0 ~~~& U \ge \alpha \\
				\alpha / U & \alpha \gamma_{\min} \le U < \alpha \\ 
				1/\gamma_{\min} & U < \alpha \gamma_{\min}
			\end{cases}
		\end{split}
	\end{align}
	
	And therefore
	\begin{align*}
		\mathbb{E}\left[ \sup_{\gamma \in (\gamma_{\min} , 1 )}   \frac{\mathds{1}\{ U \le \alpha \gamma \}}{\gamma} \right]= \alpha( 1 - \log \gamma_{\min}) . 
	\end{align*}
	Since $p_{X \rightarrow Y}^{(i)}$ is an asymptotically correct p-value, it holds that for the cdf of $p^{(i)}_{X \rightarrow Y}$ denoted by $K^{(i)}$ and the cdf of $U$ denoted by $G$ it holds that $\lim_{T\rightarrow \infty } K^{(i)}(x) \le G(x) $ for all $x$ and $i$ and therefore, for every weakly decreasing bounded function $u$ and all $i$ it holds that
	\begin{align}
		\limsup_{T \rightarrow \infty } \int u(x) ~ dK^{(i)}(x) =  \int u(x) ~ d \left(	\limsup_{T \rightarrow \infty }K^{(i)}(x) \right)  \le \int u(x) dG(x). \label{eq: mon u F G}
	\end{align}
	Seeing that for arbitrary fixed $i$, it holds
	\begin{align}
		\sup_{\gamma \in (\gamma_{\min} , 1)} \frac{\mathds{1}\{ p_{X \rightarrow Y}^{(i)} \le \alpha \gamma \}}{\gamma} =
		\begin{split} \label{eq: split p}
			\begin{cases}
				0 ~~~& p_{X \rightarrow Y}^{(i)}  \ge \alpha \\
				\alpha / p_{X \rightarrow Y}^{(i)}  & \alpha \gamma_{\min} \le p_{X \rightarrow Y}^{(i)}  < \alpha \\ 
				1/\gamma_{\min} & p_{X \rightarrow Y}^{(i)}  < \alpha \gamma_{\min},
			\end{cases}
		\end{split}
	\end{align} 
	the only difference between $\mathbb{E}\left[ \sup_{\gamma \in (\gamma_{\min} , 1 )}   \frac{\mathds{1}\{ U \le \alpha \gamma \}}{\gamma} \right]$ and $\mathbb{E}\left[ \sup_{\gamma \in (\gamma_{\min} , 1 )}   \frac{\mathds{1}\{ p_{X \rightarrow Y}^{(i)}  \le \alpha \gamma \}}{\gamma} \right]$ is in the second case in \eqref{eq: split sup U} vs \eqref{eq: split p}, but since $u(x) := \alpha / x$ is monotonically decreasing in $[\alpha \gamma_{\min}, \alpha]$ and because of \eqref{eq: mon u F G}, we obtain 
	\begin{align*}
		\limsup_{T \rightarrow \infty}\mathbb{E}\left[ \sup_{\gamma \in (\gamma_{\min} , 1 )}   \frac{\mathds{1}\{ p_{X \rightarrow Y}^{(i)}  \le \alpha \gamma \}}{\gamma} \right] \le \mathbb{E}\left[ \sup_{\gamma \in (\gamma_{\min} , 1 )}   \frac{\mathds{1}\{ U \le \alpha \gamma \}}{\gamma} \right] = \alpha (1- \log \gamma_{\min})
	\end{align*}
	Taking the mean over all panel member, we obtain
	\begin{align*}
		\limsup_{T \rightarrow \infty}	\mathbb{E}\left[ \sup_{\gamma \in (\gamma_{\min} , 1 )}   \frac{ \frac{1}{N}\sum_{i=1}^N \mathds{1}\{ p_{X \rightarrow Y}^{(i)}  \le \alpha \gamma \}}{\gamma} \right] \le (1- \log \gamma_{\min}).
	\end{align*}
	Using again the Markov inequality, we obtain
	\begin{align*}
		\limsup_{T \rightarrow \infty}	\mathbb{P}\left[ \sup_{\gamma \in (\gamma_{\min} , 1 )}   \mathds{1}\{ f_{X \rightarrow Y}(\alpha \gamma)  \ge \alpha \gamma \}\right] \le \alpha (1- \log \gamma_{\min}),
	\end{align*}
	where we use the definition of $f_{X \rightarrow Y}$ from \eqref{eq: def f}. Using \eqref{eq: equiv f q}, we obtain
	\begin{align*}
		\limsup_{T \rightarrow \infty}\mathbb{P}\left[ \inf_{\gamma \in (\gamma_{\min}, 1)} Q_{X \rightarrow Y} (\gamma) \le \alpha  \right] \le \alpha (1- \log \gamma_{\min})
	\end{align*}
	and hence
	\begin{align*}
		\limsup_{T \rightarrow \infty}\mathbb{P}\left[ \inf_{\gamma \in (\gamma_{\min}, 1)} Q_{X \rightarrow Y} (\gamma)(1- \log \gamma_{\min}) \le \alpha  \right] \le \alpha 
	\end{align*}
	which completes the proof.
\end{proof}
\section{HOW TO DEAL WITH CONFOUNDING} \label{confounding}
Since we apply Granger causality on every individual of the panel to test whether $X_i$ causes $Y_i$, we have to deal with the problem that Granger causality disregards hidden confounding. In this section, we want to give a practical solution to this problem. %\paragraph{Assumption Appendix \ref{confounding}:} \label{assump conf} For the panel there exists no bi-directional influences, \ie~if $X$ is causing $Y$, then $Y$ is not causing $X$ and visa versa. \\
In the following, we denote by $past(t)$ the past timestamps of $t$. We now give sufficient conditions for confounding vs actual causal influence. For that, we consider two arbitrary time series $W$ and $V$ without the context of panel data. The following Proposition is closely related to \cite{Mastakouri2021} Theorem 1.b. and requires partially the same assumptions which we list bellow (note that we always assume that the present cannot causally influence the past):
\paragraph{Assumptions Appendix \ref{confounding}:}
\begin{enumerate}
	\item
	The Causal Markov condition in the full time graph holds.
	\item
	Causal Faithfulness in the full time graph
	\item
	Stationary full time graph: the full time graph is invariant under a joint time shift of all variables
\end{enumerate}
\begin{prop}
	Assume that Assumption Appendix \ref{confounding} holds, $W_{past(t)} \not\independent V_t | V_{past(t)} $,  $W_{t} \independent V_{past(t)} | W_{past(t)} $ for all $t$ and that $W$ is not causing $V$. Then, there exists a (potentially high dimensional) memoryless (i.e. it does not hold that $Z_{t-1} \rightarrow Z_t$) confounder $Z$ such that the triplet $\{W, V, Z\}$ is causally sufficient and there exists at least one $t' \in past(t)$ such that $Z_{t'} \rightarrow V_t$ and $Z_{t'} \rightarrow W_{t -n}$ for some $n > 0$ but there exists no $t'' \in past(t)$ such that $Z_{t''} \rightarrow W_t$ and $Z_{t''} \rightarrow V_{t-n}$ for some $n>0$.
\end{prop}
\begin{proof}
	Assume that $W$ is not causing $V$, then according to Reichenbachs principle of common causes, there exists a confounder $Z$ between $W$ and $V$. Further, without loss of generality, we can assume causal sufficiency for the triplet $\{W, V , Z\}$ since, if another confounder exists we include it to the (potentially high dimensional) confounder $Z$. It then follows, that there exists no $t'' \in past(t)$ such that $Z_{t''} \rightarrow W_t$ and $Z_{t''} \rightarrow V_{t-n}$ for some $n>0$ because otherwise $W_{t} \not\independent V_{past(t)} | W_{past(t)} $ which is a contradiction to the assumption.
	Further, there exists at least one $t' \in past(t)$ such that $Z_{t'} \rightarrow V_t$ and $Z_{t'} \rightarrow W_{t -n}$ for some $n > 0$ because otherwise  $W_{past(t)} \independent V_t | V_{past(t)} $ which is a contradiction to the assumption.
	\\
	It remains to show that $Z$ has no memory effect. If $Z$ would have a memory effect, then \[V_t \leftarrow Z_{past(t)} \rightarrow Z_t \rightarrow Z_{t+1} \rightarrow W_{t+1}\] and hence  $W_{t+1} \not\independent V_{past(t+1)} | W_{past(t+1£)} $ which is a contradiction to the assumption.
\end{proof}
In particular, this proposition shows that if we observe $W$ causing $V$ but not $V$ causing $W$, then this can only be due to hidden confounding if the confounder has the following structure:  $Z$ has to causally influence the present of $V$ with larger time delay than the present of $W$. Further, $Z$ cannot have a memory effect.  \\
Hence, a practical way of dealing with hidden confounding is to exclude bi-directional influence between $W$ and $V$. In that way, we might decrease the true positive rate but also remove false positives due to hidden confounding. \\
Since our QPPA approach relies on Granger causality for the individual panel members, the same also holds for QPPA. \\

Here, we want to mention that \cite{Peters2017} argue after Figure 10.7.(a) that such a confounder structure occurs in real data using the example of price of butter and cheese. They state that the price of butter and cheese are confounded by the price of milk but the influence from milk to cheese has a larger time delay than from milk to butter because it takes longer to produce cheese. \\ However, this is only an example for such a confounder structure if the price of milk has no memory effect, otherwise we observe significant influence (of course depending on the strength of influence and memory effect) from the past of cheese to the present of butter \textbf{and} from the past of butter to the present of cheese. 
\section{DETECTING INSTANTANEOUS CAUSAL INFLUENCES} \label{instantenous effects}
According to the dynamics \eqref{eq:gen-proc}, instantaneous effects are excluded. However, as \cite{Peters2017} Chapter 10.3 argue that, under some assumptions, Granger causality is able to detect causal influence even if it is purely instantaneous. Namely, if faithfulness holds and time series $X$ has a memory effect (note that here we specifically consider the bi-variate case, in multi variate cases, instantaneous effects can lead to non-identifiability), \textit{i.e.}, $X_{t-1} \rightarrow X_t$ for all $t$, then if $X$ causes $Y$, the present of $Y$ is not independent of the past of $X$ given the past of $Y$ because of the causal influence $X_{t-1}  \rightarrow X_t \rightarrow Y_t$. The following figure, taken from \cite{Peters2017}  illustrates this dependence.
\begin{figure}[ht]
	\centering
	\includegraphics[scale=0.6]{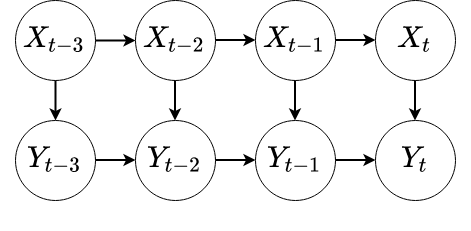}
	\caption{\textit{Example of a causal relation with only instantaneous influences taken from \cite{Peters2017} Figure 10.8.(b) where Granger causality is able to detect the causal relation from $X$ to $Y$. In this example, $X$ and $Y$ denote single time series, contrary to our panel data notation where $X$, $Y$ denote panels.} } \label{fig: inst effect}
\end{figure}
Hence, including the past of $X$ into the prediction of $Y$ decreases the prediction error and therefore Granger causality could be detected. Since our panel Granger Non-causality test relies on Granger causality on every individual panel member, our algorithm is thus capable of detecting causal influence even if it is purely instantaneous. This is shown in the following experiment.
\paragraph{Experiment for instantaneous effects:}  The generation process with only instantaneous causal effects (which is, as the processes in Section \ref{sec: experiments}, an autoregressive process of order 1) is specified by:
\begin{align*}
	X_{i,t} &= \delta_{i, 1} X_{i, t-1} + \eta_{i, t}, \\
	Y_{i, t} &= \theta_{i, 1} Y_{i, t-1} + \beta X_{i, t} +  \epsilon_{i, t},
\end{align*}
where the innovation processes are i.i.d. Gaussian random variables with $\eta_{i, t}, \epsilon_{i, t} \sim N(0, 0.1)$ and i.i.d. $\delta_{i, 1}, \theta_{i, 1} \sim Unif(0.2, 0.8)$. Further, we either draw $\beta $ from $Unif(0.2, 0.8)$ if the null should be rejected or set $\beta = 0$ if the null should be accepted. The instantaneous effect comes from the influence $ \beta X_{i, t}$ on $Y_{i, t}$ and this generation process results in the causal structure corresponding to Figure \ref{fig: inst effect}. Since the dependence of $X_{t-1}$ and $Y_t$ given $Y_{t-1}$ is indirect through the memory effect and the instantaneous effect, we need a stronger signal to detect the relation via QPPA, therefore we also include results of QPPA where we draw $\beta$ from $Unif(0.6, 0.8)$.
The results in Figure \ref{fig: inst diff gamma} show that for $T=50$, the power decreases relatively fast. Also, for $T=100$, the power is smaller than for causal influence with time delay. However, we also see in in Figure \ref{fig: inst diff gamma}.(d) that QPPA is still able to recover the true causal relation in most cases for $\gamma$ between $0.01$ and approx.~$0.5$ even if the influence is purely instantaneous.
\begin{figure}[ht]
	\begin{subfigure}{.5\textwidth}
		\centering
		\includegraphics[scale=0.4]{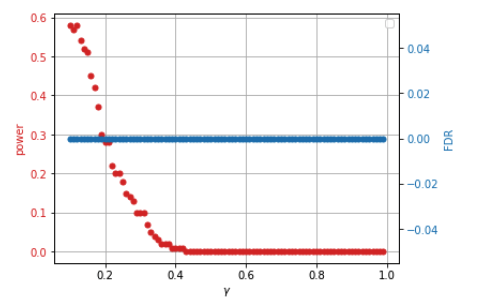} \label{fig:a covid diff gamma 30}
		\caption{$T=50, ~\beta \sim Unif(0.2, 0.8)$}
	\end{subfigure}
	\begin{subfigure}{.5\textwidth}
		\centering
		\includegraphics[scale=0.4]{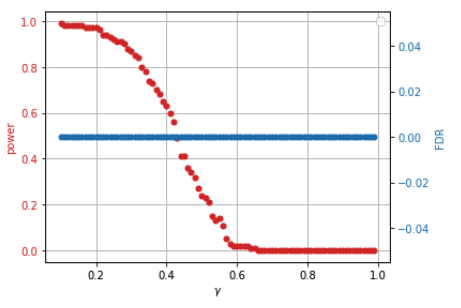} \label{fig:a covid diff gamma 60}
		\caption{$T=100, ~ \beta \sim Unif(0.2, 0.8$}
	\end{subfigure}
	\begin{subfigure}{.5\textwidth}
		\centering
		\includegraphics[scale=0.4]{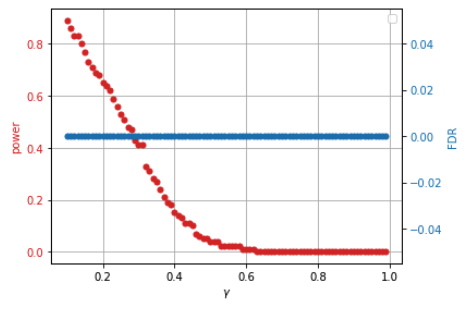} \label{fig:a covid diff gamma 30}
		\caption{$T=50, ~ \beta \sim Unif(0.6, 0.8)$}
	\end{subfigure}
	\begin{subfigure}{.5\textwidth}
		\centering
		\includegraphics[scale=0.4]{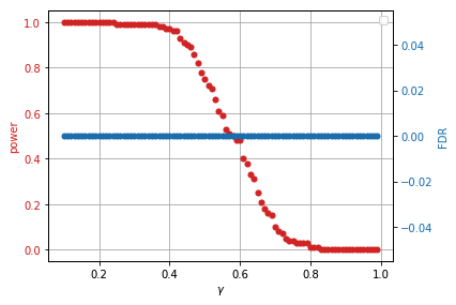} \label{fig:a covid diff gamma 60}
		\caption{$T=100, ~ \beta \sim Unif(0.6, 0.8)$}
	\end{subfigure}
	\caption{\textit{Empirical results of Experiment for instantaneous effects for QPPA with $\gamma = 0.01, 0.02, \dots , 0.99$. We choose $N=30$ and $T=50, 100$.}} \label{fig: inst diff gamma}
\end{figure}
\section{DEALING WITH NON-STATIONARITIES} \label{non-stationary}
A common way to deal with non-stationarities to test for Granger Non-causality is to difference the processes. More precisely, following \cite{Granger81, Engle87}, we say that a time series $W$ is integrated of order $d$, denoted by $I(d)$, if  $\{(1-L)^d W_t\}_t$ is (covariance) stationary, where $L$ denotes the lag operator. 
Different test procedures can be applied to find the order of integration, we examine one option for the COVID-19 study in Appendix 
\ref{sec: covid preprocessing}.  \\
If we want to test the existence of Granger causality between two time series $W$ and $V$, we can first search for the order of integration using stationarity tests on the $n$-th difference process where we stop as soon as stationarity gets accepted. If $W$ and $V$ have different orders of integration, say $d_W$ and $d_V$, we take the maximum and difference both time series $\max(d_W, d_V)$-times. The difference processes often lead to stationary time series in practice. However, note that such a $d$ does not necessarily exist and hence this approach is not applicable for all datasets. 
\section{FURTHER EXPERIMENTS WITH SYNTHETIC DATA} \label{sec: further exp synth}
\paragraph{Experiment for sporadically missing connections:} For this experiment we again consider the model without cross-sectional dependencies:
\begin{align*}
	X_{i,t} &= \delta_{i, 1} X_{i, t-1} + \eta_{i, t}, \\
	Y_{i, t} &= \theta_{i, 1} Y_{i, t-1} + \beta X_{i, t-1} +  \epsilon_{i, t},
\end{align*} 
where the innovation processes are i.i.d. Gaussian random variables with $\eta_{i, t}, \epsilon_{i, t} \sim N(0, 0.1)$ and i.i.d. $\delta_{i, 1}, \theta_{i, 1} \sim Unif(0.2, 0.8)$. Further, we either draw $\beta $ from $Unif(0.2, 0.8)$ if the null should be rejected or set $\beta = 0$ if the null should be accepted. To show robustness against sporadically missing connections, we set $\beta$ to zero with probability $a$ in the case where the null should be rejected where we let $a$ range from $0.1$ to $0.9$. The proportion of missing connections when the null should be rejected is $a$, \textit{i.e.}, for every panel member there is the chance of $a$ that $\beta = 0$ although the null should be rejected.
\begin{figure}[ht]
	\begin{subfigure}{.5\textwidth}
		\centering
		\includegraphics[scale=0.4]{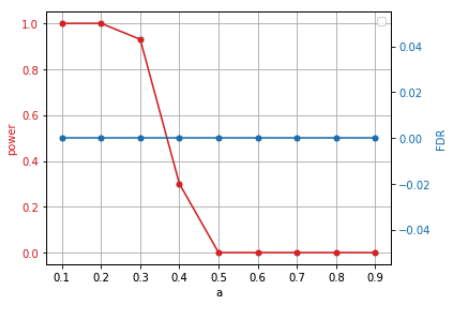} \label{fig:a covid diff gamma 30}
		\caption{$T=50$}
	\end{subfigure}
	\begin{subfigure}{.5\textwidth}
		\centering
		\includegraphics[scale=0.4]{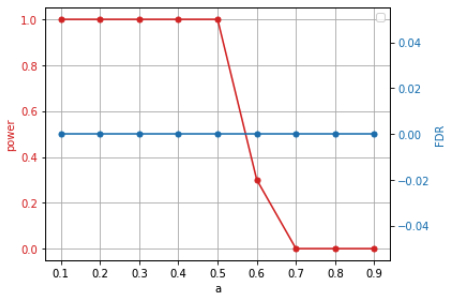} \label{fig:a covid diff gamma 60}
		\caption{$T=100$}
	\end{subfigure}
	\caption{\textit{Empirical results of Experiment for sporadically missing edges for QPPA, where we randomly set $\beta$ to $0$ in the case where the null should be rejected with probability $a$ where we let $a$ range from $0.1$ to $0.9$. We choose $N=100$ and $T=50, 100$ and $\gamma = 0.5$.}} \label{fig: spor miss diff gamma}
\end{figure}
The results in Figure \ref{fig: spor miss diff gamma} show that especially in the large sample regime, QPPA is robust against sporadically missing edges since the power drops only after the probability that an edges is missing although it should be there is higher than $50\%$.
\paragraph{Experiment non-Gaussian noise:} Again, we consider the model without cross-section dependencies:
\begin{align*}
	X_{i,t} &= \delta_{i, 1} X_{i, t-1} + \eta_{i, t}, \\
	Y_{i, t} &= \theta_{i, 1} Y_{i, t-1} + \beta X_{i, t-1} +  \epsilon_{i, t},
\end{align*} 
where the innovation processes are i.i.d. \textbf{uniformly distributed} random variables with $\eta_{i, t}, \epsilon_{i, t} \sim Unif(0.2, 0.8)$ and i.i.d. $\delta_{i, 1}, \theta_{i, 1} \sim Unif(0.2, 0.8)$ and $\beta \sim Unif(0.2, 0.8)$ if the null should be rejected and $\beta = 0$ if the null should be accepted.
\begin{figure}[ht]
	\begin{subfigure}{.5\textwidth}
		\centering
		\includegraphics[scale=0.37]{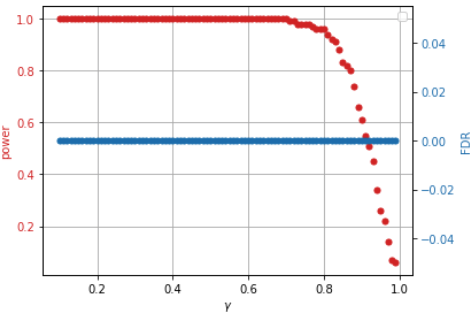} \label{fig:a covid diff gamma 30}
		\caption{$T=50$}
	\end{subfigure}
	\begin{subfigure}{.5\textwidth}
		\centering
		\includegraphics[scale=0.4]{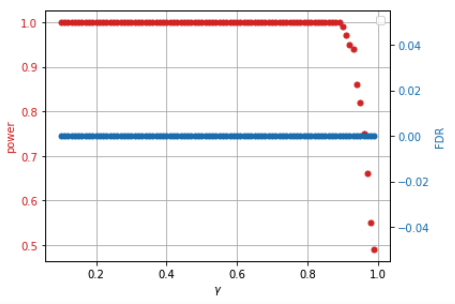} \label{fig:a covid diff gamma 60}
		\caption{$T=100$}
	\end{subfigure}
	\caption{\textit{Empirical results of experiment for non-Gaussian noise effects for QPPA with $\gamma = 0.01, 0.02, \dots , 0.99$. We choose $N=30$ and $T=50, 100$.}} \label{fig: non-gauss diff gamma}
\end{figure}
The results in Figure \ref{fig: non-gauss diff gamma} show that non-Gaussian noise does not decrease the power/FDR of QPPA, \textit{cf.} Figure \ref{fig: varying gamma exp 1}. 
\paragraph{Experiment non i.i.d noise:} Also for this experiment we use the model without cross-sectional dependencies:
\begin{align*}
	X_{i,t} &= \delta_{i, 1} X_{i, t-1} + \eta_{i, t}, \\
	Y_{i, t} &= \theta_{i, 1} Y_{i, t-1} + \beta X_{i, t-1} +  \epsilon_{i, t},
\end{align*} 
where we sample the innovation processes from 3-blocks, \textit{i.e.}, $\eta_{i, 1}, \eta_{i, 2}, \eta_{i, 3}$ are dependent, $\eta_{i, 4}, \eta_{i, 5}, \eta_{i, 6}$ are dependent etc. and similar for $\epsilon_{i, t}$. Further, $\delta_{i, 1}, \theta_{i, 1} $ are  i.i.d. $Unif(0.2, 0.8)$

\begin{figure}[ht]
	\begin{subfigure}{.5\textwidth}
		\centering
		\includegraphics[scale=0.3]{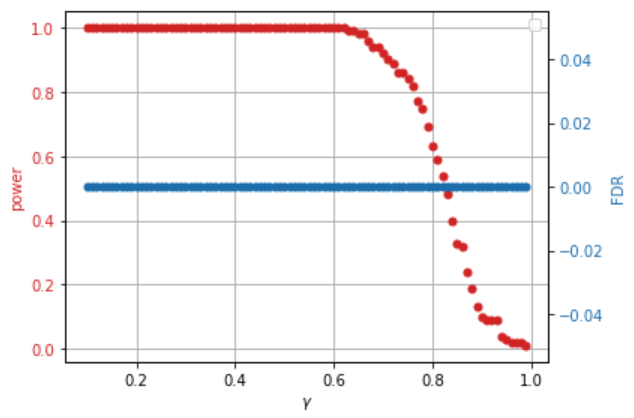} \label{fig:a covid diff gamma 30}
		\caption{$T=50$}
	\end{subfigure}
	\begin{subfigure}{.5\textwidth}
		\centering
		\includegraphics[scale=0.3]{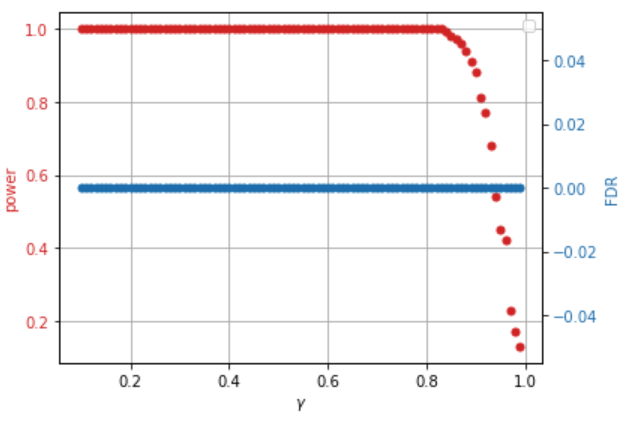} \label{fig:a covid diff gamma 60}
		\caption{$T=100$}
	\end{subfigure}
	\caption{\textit{Empirical results of Experiment for non-Gaussian noise effects for QPPA with $\gamma = 0.01, 0.02, \dots , 0.99$. We choose $N=30$ and $T=50, 100$.}} \label{fig: dependent noise diff gamma}
\end{figure}
In Figure \ref{fig: dependent noise diff gamma}, we see similar results as in Figure \ref{fig: non-gauss diff gamma}. It shows that QPPA is robust against this type of dependent residuals.
\section{ADDITIONAL RESULTS TO COVID STUDY}
\subsection{COVID-19 Data Preprocessing Steps} \label{sec: covid preprocessing}
The preprocessing steps we apply on the COVID-19 data about confirmed cases and deaths are as follows: 
\begin{enumerate}
	\item We remove all members of the panel (\textit{i.e.}, Countries/Regions) which have at least one missing value in confirmed cases or deaths since neither the Granger causality implementation of \texttt{statsmodel} nor \texttt{xtg-cause} can deal with missing values. However, note that QPPA could deal with missing values here, it would only require an implementation of Granger causality that can deal with missing values.  
	\item Although the dataset contains records starting from 22nd January 2020, we only consider the time range 1st November 2020 until 4th October 2021 because many countries did not have confirmed COVID-19 cases (and therefore also no death cases) in the beginning of the pandemic. 
	\item We standardize the data.  
	\item Since Granger causality can only deal with stationary time series, we first apply a stationarity test for panel data that we developed based on the augmented Dickey-Fuller test (adf test) \citep{dickey1979}, combined with the p-value aggregation method explained in Section \ref{sec::p-values for high-dimensional regression} to test for unit root for each member of the panel. We include more explanation and the test results in Appendix \ref{sec: panel adf}.
	If the null gets accepted to the significance level of 5\% (where the null is that there exists a unit root, \textit{i.e.}, the time series are non-stationary), we generate the first difference of the panel members and apply the test again on the difference process. If the test again accepted the null (again to the significance level of 5\%), we generate the second order difference processes. We continue this procedure until the null gets rejected. The results are also given in Appendix \ref{sec: panel adf}. Our test procedure rejects the null for the second order difference process which we will use from now on for the following analysis. 
	\item From the panel we remove those members for which either the stationary version of confirmed cases or deaths is constant since causality cannot be conducted from constant time series. 
\end{enumerate}
\subsection{Cross-Sectional Dependence Test} \label{sec: covid panel dependence}
To test for cross-sectional dependence, we use the \texttt{Stata}-implementation \citep{Dehoyos06} which contains the test procedures of \cite{Pesaran04, Frees95, Friedman37}. There, we apply every test on the preprocessed data according to Appendix \ref{sec: covid preprocessing}.
The null of these tests is that the individual panel members are independent. For every test, we obtain the p-value 0.000 and therefore all three tests strongly reject the null. Hence, we accept the alternative that cross-sectional dependencies exist.
\subsection{Non-Stationarity Test} \label{sec: panel adf}
There are several existing non-stationarity tests for panel data, see \cite{Breitung00, Breitung05, Choi01, Hadri00, Harris99, Im03, Levin02}. However, most of them assume independence across panel members. Here, we want to present another panel stationarity test that relies on the same procedure as our QPPA approach, Section \ref{sec:panel_boot}, except that in step 1 we instead of applying Granger causality on the individual panel members, we apply a stationarity test and then aggregate the corresponding p-values with the procedure described in step 2. For the stationarity test, we use the augmented Dickey-Fuller test (adf), see \cite{dickey1979}. Hence, we apply the adf test on each individual panel member on the COVID-19 data about confirmed cases and deaths, where we use 12 time lags.
Afterwards, we aggregate these p-values using the aggregation method 
\begin{align*}
	\min\{1, \text{emp. } \gamma \text{-quantile}  (p_j / \gamma : j \in 1, \dots , N) \},
\end{align*}
where $p_j$ denotes the p-value of the adf test for the $j$-th panel member. Our hypothesis test reads
\begin{align*}
	H_0&: \text{The panel has a unit root} \\
	H_1&: \text{The panel has no unit root}.
\end{align*}
Similar to the hypotheses test of QPPA, we assume that either each member of the panel has a unit root or none of them. \\
To find the order of integration (see Section \ref{non-stationary}), we apply the adf test combined with the second step of QPPA.
Table \ref{fig: adf test} show the results  for different $\gamma$. We see consistent rejection of the null for the second order difference processes, whereas for the first order we only reject for $\gamma = 0.1$ and without generating the difference process, the null is not rejected for any $\gamma \in \{0.1, 0.25, 0.5, 0.75, 0.9\}$. Hence, we have clear indication that the COVID-19 data is second difference order stationary. Since we need stationary data for Granger causality, we will henceforth use the second order processes for the analysis.
\begin{table} [ht]
		\caption{\textit{Empirical results for the stationarity test of COVID-19 using adf combined with the second step of QPPA. With a slight abuse of notation, we say that the panel is $d$-order integrated (denoted by $I(d)$) if the null of our panel adf test is rejected for the panel $\{\{(1-L)^d X_{i, t}\}_t\}_i$ and $\{\{(1-L)^d Y_{i, t}\}_t\}_i$ respectively where $X$ denotes confirmed cases and $Y$ denotes deaths, where $L$ denotes the back shift operator. For more details see Section \ref{non-stationary}.}} \label{fig: adf test}
	\centering
		\begin{tabular}{clll}
			\toprule
			Order of  &  $\gamma$    &  \multicolumn{2}{c}{p-val} \\
			integration		&      & confirmed cases &      deaths \\
			\midrule
			$I(0)$ & 0.1 &  0.414 & 1.0 \\
			& 0.25 & 1.0  &      1.0 \\
			& 0.5 & 1.0  &  1.0 \\
			& 0.75 &  1.0 &  1.0  \\
			& 0.9 &  1.0 &  1.0     \\
			$I(1)$ & 0.1 &  0.001 & 0.003  \\
			& 0.25 &    0.337  & 0.098   \\
			& 0.5 &  0.569 &  0.426 \\
			& 0.75 &  0.745 &  0.641  \\
			& 0.9 &0.957    &  0.843     \\
			$I(2)$ & 0.1 &9.883e-14 &  1.877e-12  \\
			& 0.25 & 2.297e-12  & 2.457e-10       \\
			& 0.5 & 2.543e-08  &  1.702e-06 \\
			& 0.75&5.574e-05    &  0.001  \\
			& 0.9&0.004    & 0.006     \\
			
			\bottomrule
		\end{tabular}
\end{table}
\newpage
\subsection{Additional COVID-19 Experiments} \label{sec: additional covid exp}
We repeat the experiment in Figure \ref{fig: cov diff gamma} with different number of countries/regions. The results, given in Figure \ref{fig:appendix varying gamma covid}, show that the accuracy of QPPA increases with increasing number of countries/regions per run.
\begin{figure}[h]
	\begin{subfigure}{.33\textwidth}
		\includegraphics[scale=0.3]{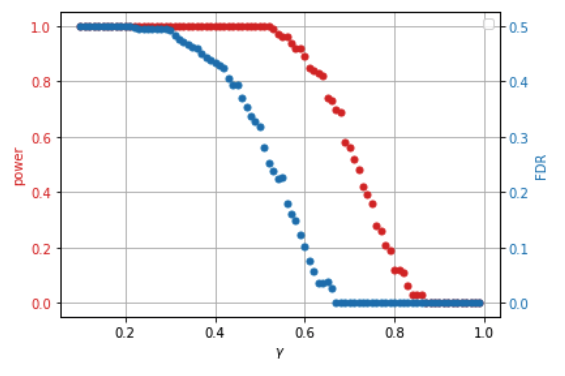} \label{fig:a covid diff gamma 30}
		\caption{$N=30$}
	\end{subfigure}
	\begin{subfigure}{.33\textwidth}
		\includegraphics[scale=0.3]{images/covid_diff_gamma_60.png} \label{fig:a covid diff gamma 60}
		\caption{$N=60$, also used in the main text}
	\end{subfigure}
	\begin{subfigure}{.33\textwidth}
		\includegraphics[scale=0.3]{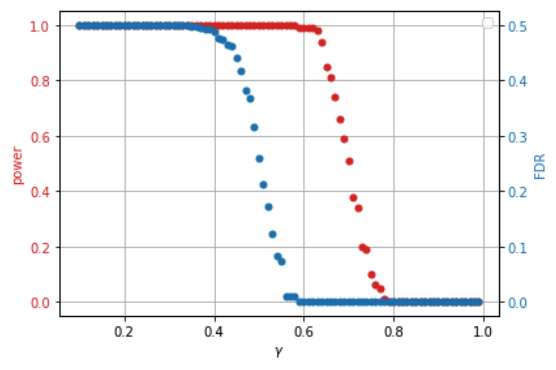} \label{fig:a covid diff gamma 100}
		\caption{$N=100$}
	\end{subfigure}
	\caption{Empirical results for COVID-19 data about confirmed cases and deaths using QPPA with $\gamma$ = 0.01, 0.02, . . . , 0.99. To calculate power and FDR, we randomly selected $N$ countries/regions out of the 225 (here, $N$ is specified in the subfigures) and checked whether QPPA detects the causal relation c → d and d → c respectively to the significance level 5\% and repeat this 100 times. Here, we use the complete time series length available which is $T=335$.
	} \label{fig:appendix varying gamma covid}
\end{figure}
\end{document}